\documentclass[
    reprint,
    prxquantum,
    superscriptaddress,
    nofootinbib,
    amsmath,amssymb,
    aps
]{revtex4-2}

\usepackage[colorlinks=true,linkcolor=blue,citecolor=blue,urlcolor=blue]{hyperref}
\usepackage{physics}
\usepackage{amsthm}
\usepackage{enumitem}
\usepackage{here}
\usepackage[legacycolonsymbols]{mathtools}

\newtheorem{thm}{Theorem}
\newtheorem*{thm*}{Theorem}
\newtheorem{cor}{Corollary}
\newtheorem*{cor*}{Corollary}
\newtheorem{dfn}{Definition}
\newtheorem*{dfn*}{Definition}
\newtheorem{lem}{Lemma}
\newtheorem*{lem*}{Lemma}
\newtheorem{prop}{Proposition}

\begin{document}

\title{Information Erasure and Quantum Imprint in Quantum Measurement \\and First-Order SPAM Error Separation}

\author{Taiga Suzuki}%
\affiliation{%
Department of Physics, Institute of Science Tokyo,
Meguro-ku, Tokyo 152-8551, Japan
}%

\author{Yuki Ito}%
\affiliation{%
Graduate School of Engineering Science, The University of Osaka,
1-3 Machikaneyama, Toyonaka, Osaka 560-8531, Japan
}%

\author{Masayuki Ohzeki}%
\affiliation{%
Department of Physics, Institute of Science Tokyo,
Meguro-ku, Tokyo 152-8551, Japan
}%
\affiliation{%
Graduate School of Information Sciences, Tohoku University,
Sendai, Miyagi 980-8579, Japan
}%
\affiliation{%
Research and Education Institute for Semiconductors and Informatics,
Kumamoto University, Kumamoto 860-8555, Japan
}%
\affiliation{%
Sigma-i Co., Ltd., Minato-ku, Tokyo 108-0075, Japan
}%

\date{\today}%
\begin{abstract}
    We introduce information erasure and quantum imprint as two properties that classify quantum instruments.
    Information erasure is the property that an appropriate postselection can render the distribution of
    earlier measurement outcomes independent of the initial quantum state while retaining all outcome
    branches.
    Quantum imprint is the complementary property that no admissible postselection can eliminate this state
    dependence.
    We show that this classification has a nontrivial structure and that the natural intuition that
    measurements providing more information about the initial quantum state should be less likely to exhibit
    information erasure does not hold in general.

    We further show that, under a sufficiently reliable postselection, information erasure enables first-order
    separation of state-preparation and measurement (SPAM) errors.
    Specifically, the first-order contribution of state-preparation error vanishes from the posterior
    distribution, whereas visible first-order contributions of measurement error remain.
    This result recasts SPAM error separation from the problem of simultaneously characterizing state
    preparation and measurement into the problem of realizing a reliable postselection.
\end{abstract}

\maketitle
\section{introduction}
Quantum measurement provides an essential interface between the quantum world and the classical information
accessible to us.
Understanding quantum measurement is therefore important not only for the foundations of quantum mechanics,
but also for essentially all aspects of quantum information processing.

Among the various techniques used in quantum measurement, postselection—in which only those experimental
trials yielding a specified outcome in a subsequent measurement are retained—has been widely studied as a
means of modifying the statistics of an earlier measurement~\cite{AharonovBergmannLebowitz1964,AharonovVaidman1991,AharonovPopescuTollaksen2010}.
Its effects have been investigated in a variety of contexts, including weak values, postselected inference,
and the general structure of sequential quantum measurements~\cite{AharonovAlbertVaidman1988,KofmanAshhabNori2012,Dressel2014WeakValues,Fritz2010,DresselJordan2012,DresselJordan2013Instruments,Silva2014PrePostselected,Kiktenko2023Postselection,PinskeMolmer2025Retrodiction}.

A prominent use of postselection is to enhance the sensitivity of observed statistics to a particular physical
quantity, as exemplified by weak-value amplification, Fisher-information concentration, and postselected
metrological protocols~\cite{Dixon2009,Gendra2013Abstention,Jordan2014TechnicalAdvantages,ArvidssonShukur2020QuantumAdvantage}.
Here, we focus on the converse possibility: can postselection completely eliminate the sensitivity of an
earlier measurement outcome to the initial quantum state?

In this work, we show that the answer to this question is sharply divided into affirmative and negative cases
according to the structure of the quantum instrument.
For one class of quantum instruments, an appropriate postselection can completely eliminate the dependence of
the posterior distribution of the earlier measurement outcome on the initial quantum state.
For another class, by contrast, no postselection can completely eliminate this state dependence while
retaining all branches of the earlier measurement outcomes.
We refer to the former property as \emph{information erasure} and the latter as \emph{quantum imprint}, and
classify quantum instruments according to these two properties.
Fig.~\ref{instrument_partition} schematically illustrates this classification.

Information erasure and quantum imprint may, at first sight, appear to be closely related to how much
information the first measurement acquires about the quantum state $\rho$.
We show that this natural intuition fails in two respects.
First, even two quantum instruments implementing exactly the same POVM can exhibit opposite behaviors, with
one admitting information erasure and the other exhibiting quantum imprint.
Second, information erasure is compatible with informational completeness.
That is, even a quantum instrument that induces information erasure can implement an informationally complete
POVM before postselection.
Therefore, information erasure does not arise simply because the first measurement acquires little information
about the initial state.

Beyond this structural result, we show that information erasure has a direct operational consequence in the
presence of experimental errors.
We introduce unknown small perturbations in both an instrument exhibiting information erasure and the input
quantum state, and expand the posterior distribution to first order.
The first-order contribution of state-preparation error then cancels exactly, while visible first-order
measurement errors of the instrument remain.
The difficulty of distinguishing state-preparation and measurement (SPAM) errors is a central problem in
quantum-device characterization, and various approaches have been proposed to address this intrinsic ambiguity~\cite{Nielsen2021GST,Cattaneo2023SelfConsistentQMT,Jayakumar2024Universal,Lin2021IndependentSPAM,Yu2025SeparateSPAM,Stark2014,JacksonVanEnk2015}.

Our approach enables first-order SPAM error separation under two assumptions: the instrument exhibits
information erasure, and the postselection can be implemented with sufficient reliability, meaning that its
imperfections do not affect the posterior distribution at first order.
First, as follows from one of our main results, for some POVMs, the same POVM admits instrument realizations
exhibiting either information erasure or quantum imprint, and instruments exhibiting information erasure can
also implement informationally complete POVMs.
Thus, requiring information erasure does not restrict our approach to special measurements with limited
information-acquisition capability.
Second, sufficiently reliable postselection allows the SPAM error separation problem to be reframed.
Instead of simultaneously characterizing state preparation and measurement with high precision, one need only
realize a sufficiently reliable postselection.

These results show that state dependence in quantum measurement not only characterizes a nontrivial structure
of measurement processes, but can also serve as an experimentally exploitable resource.

The remainder of this paper is organized as follows.
In Sec.~II, we introduce quantum instruments and postselection, define information erasure and quantum
imprint, and present important classes exhibiting these properties together with their structural
characterization.
In Sec.~III, we discuss first-order SPAM error separation based on information erasure and characterize the
observable components of measurement error.
Finally, in Sec.~IV, we summarize our results and discuss future directions.

\begin{figure}[h]
    \includegraphics[width=6cm]{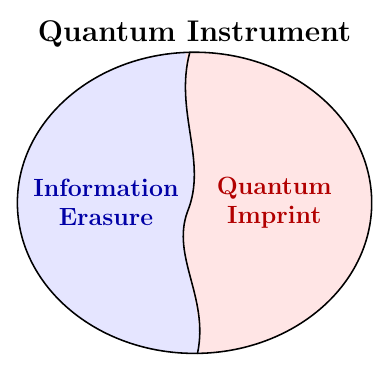}
    \caption{Schematic classification of quantum instruments according to information erasure and quantum
    imprint.
    The curved boundary is schematic and emphasizes the nontrivial structure of this classification.}
    \label{instrument_partition}
\end{figure}

\section{Information Erasure and Quantum Imprint}
In this section, we introduce the concepts of information erasure and quantum imprint, identify several
important classes of quantum instruments exhibiting each property, and establish structural results showing
that this dichotomy is nontrivial and can lead to counterintuitive behavior.

We first introduce the sequential measurement setting in which these two concepts are defined.
Throughout, $\mathcal{H}$ denotes a finite-dimensional Hilbert space, $X$ and $Y$ denote finite outcome sets,
and $\mathcal{S}(\mathcal{H})$ denotes the set of quantum states on $\mathcal{H}$.
As illustrated in Fig.~\ref{fig_diagram}, an input state $\rho\in\mathcal{S}(\mathcal{H})$ is first measured
by a quantum instrument~\cite{Ozawa1984,Ozawa1985}
\begin{align}
    \mathcal{I} = \{\mathcal{I}_x\}_{x\in X},
\end{align}
where each $\mathcal{I}_x$ is completely positive and $\sum_{x\in X}\mathcal{I}_x$ is trace preserving.
The probability of obtaining the first outcome $x$ is
\begin{align}
    p(x|\rho) = \operatorname{Tr}\!\left[\mathcal{I}_x(\rho)\right].
\end{align}

A second measurement, described by a POVM $\{P_y\}_{y\in Y}$, is subsequently performed.
The joint probability of the two outcomes is
\begin{align}
    p(x,y|\rho) = \operatorname{Tr}\!\left[P_y\,\mathcal{I}_x(\rho)\right] = \operatorname{Tr}\!\left[\rho\,\mathcal{I}_x^\dagger(P_y)\right].
\end{align}
$\mathcal{I}_x^\dagger$ denotes the dual map of $\mathcal{I}_x$.
We postselect the runs in which the second measurement yields a fixed outcome $y^\ast\in Y$.
For any state and measurements satisfying
${\operatorname{Tr}\!\left[\rho\sum_{z\in X}\mathcal{I}_z^\dagger(P_{y^\ast})\right]}>0$, the conditional
distribution of the earlier outcome can be calculated in the same way as~\cite{DresselJordan2012,Gammelmark2013}:
\begin{align}
    p(x|y^\ast,\rho) = \frac{p(x,y^\ast|\rho)}{\sum_{x\in X}p(x,y^\ast|\rho)} = \frac{\operatorname{Tr}\!\left[\rho\,\mathcal{I}_x^\dagger(P_{y^\ast})\right]}{\operatorname{Tr}\!\left[\rho\sum_{z\in X}\mathcal{I}_z^\dagger(P_{y^\ast})\right]}.
    \label{eq:postselected_distribution}
\end{align}
We refer to $p(x|y^\ast,\rho)$ as the posterior distribution of the earlier measurement outcome.
Our starting point is to ask how postselection can transform this posterior distribution with respect to the
input quantum state.
As we show in the next subsection, this question leads to a dichotomy between two classes of instruments.

\subsection{Conditions for Information Erasure and Quantum Imprint}
\label{subsec:information_erasure_quantum_imprint}
Using the sequential measurement setting introduced above, we now define information erasure and quantum
imprint.

\begin{figure*}[t]
    \centering
    \includegraphics[width=0.8\textwidth]{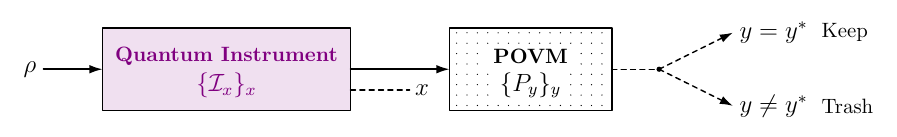}
    \caption{Schematic of the sequential measurement and postselection protocol.
    An input state $\rho$ is first measured by a quantum instrument $\{\mathcal{I}_x\}_{x\in X}$, yielding
    outcome $x$, followed by a POVM $\{P_y\}_{y\in Y}$.
    Only the runs yielding the outcome $y=y^\ast$ are retained.
    The central question is whether this postselection can eliminate the dependence of the earlier outcome
    distribution on the input state $\rho$ while retaining all first measurement outcome branches.}
    \label{fig_diagram}
\end{figure*}
Let $\Delta_X^\circ:=\{q=(q_x)_{x\in X}| q_x>0\ \forall x\in X,\sum_{x\in X}q_x=1\}$ denote the set of
probability distributions with full support on $X$.

Furthermore, to restrict attention to physically possible postselections, we call a POVM effect $P$ admissible
for the instrument $\mathcal{I}=\{\mathcal{I}_x\}_{x\in X}$ if the corresponding postselection outcome occurs
with nonzero probability for at least one input state.

We denote the set of all such admissible postselection effects by

\begin{align}
    \operatorname{Adm}(\mathcal{I}) := \left\{P\in\mathcal{L}(\mathcal{H}) \,\middle|\, 0\leq P\leq I,\ \sum_{z\in X}\mathcal{I}_z^\dagger(P)\neq 0\right\}.
    \label{physical_realization}
\end{align}
We therefore define information erasure as follows:
\begin{dfn}[Information-erasure-inducing instrument]
    \label{dfn:information_erasure_inducing_instrument}
    A quantum instrument $\mathcal{I}=\{\mathcal{I}_x\}_{x\in X}$ is called an
    \emph{information-erasure-inducing instrument} if
    \begin{equation}
        \begin{aligned}
            \exists\,P_{y^\ast}\in \operatorname{Adm}(\mathcal{I}),\ \exists\,q\in\Delta_X^\circ\quad \text{such that} \\
            \forall x\in X,\quad \mathcal{I}_x^\dagger(P_{y^\ast}) = q_x\sum_{z\in X}\mathcal{I}_z^\dagger(P_{y^\ast})
        \end{aligned}
        \label{eq:information_erasure_operator}
    \end{equation}
    When Eq.~\eqref{eq:information_erasure_operator} holds, the pair $(\mathcal{I},P_{y^\ast})$ is said to
    realize \emph{information erasure}.
\end{dfn}
To connect this operator definition with the probability distribution, consider a quantum state
$\rho\in \mathcal{S}(\mathcal{H})$, suppose that $p(x|\rho)>0$ and depends on $\rho$.
If $(\mathcal{I},P_{y^\ast})$ realizes information erasure, then
\begin{align}
    p(x| y^\ast) = q_x
\end{align}
holds and since $q_x$ is independent of $\rho$, $p(x| y^\ast)$ no longer depends on $\rho$.
The requirement $q_x>0$ for every $x\in X$ excludes branch elimination, ensuring that state independence is
realized among all outcomes of the first measurement rather than by postselection discarding some of the
branches entirely.

Moreover, information erasure should also be distinguished from measurement reversal or quantum uncollapsing~\cite{UedaImotoNagaoka1996,KoashiUeda1999,KorotkovJordan2006,Katz2008Reversal,Kim2009Reversal,JordanKorotkov2010}.
In measurement-reversal protocols, a recovery operation is typically designed to undo the state disturbance
associated with a particular measurement outcome and may therefore depend explicitly on that outcome.
By contrast, in our setting the postselection effect $P_{y^\ast}$ is fixed independently of the earlier
outcome $x$, and the objective is not to restore the premeasurement state.

Interestingly, one may also consider the complementary scenario: are there quantum instruments for which no
choice of postselection can erase the state dependence of the earlier measurement record?
The answer is yes, and we show that such instruments indeed exist.
For these instruments, no admissible postselection can make the earlier measurement record state independent
while retaining full support on $X$.

We refer to this as \emph{quantum imprint} as defined below.
\begin{dfn}[Quantum imprint]
    \label{dfn:quantum_imprint}
    A quantum instrument $\mathcal{I}=\{\mathcal{I}_x\}_{x\in X}$ is called a \emph{quantum imprint} if
    \begin{equation}
        \begin{aligned}
            \forall\,P_{y^\ast}\in \operatorname{Adm}(\mathcal{I}),\ \forall\,q\in\Delta_X^\circ \\
            \exists\,x\in X\quad \text{such that}\quad\mathcal{I}_x^\dagger(P_{y^\ast})\neq q_x\sum_{z\in X}\mathcal{I}_z^\dagger(P_{y^\ast}).
        \end{aligned}
        \label{eq:quantum_imprint_operator}
    \end{equation}
\end{dfn}
The probabilistic meaning is that, for any postselection $P_{y^\ast}\in \operatorname{Adm}(\mathcal{I})$,
there exists no $q\in\Delta_X^\circ$ such that $p(x|y^\ast)=q_x$ for all input states $\rho$ and all $x\in X$.

Operationally, the term ``quantum imprint'' reflects the fact that the state dependence left in the earlier
measurement record cannot be erased by any admissible postselection while retaining all outcome branches.

These two definitions provide a mutually exclusive and exhaustive classification of quantum instruments.
Indeed, Eq.~\eqref{eq:quantum_imprint_operator} is precisely the logical negation of
Eq.~\eqref{eq:information_erasure_operator}.
Thus, every quantum instrument belongs to exactly one of the two classes.

In the next subsection, we characterize representative classes of quantum instruments exhibiting information
erasure and quantum imprint and clarify the structure underlying this classification.
\subsection{Structural Results}
\label{sec:structural_results}

We now present representative classes of instruments that realize information erasure and quantum imprint.
Throughout this section, for simplicity, we restrict our attention to single-Kraus instruments of the form
\begin{align}
    \mathcal{I}_x(\cdot) = M_x(\cdot)M_x^\dagger,\qquad\sum_{x\in X}M_x^\dagger M_x = I.
\end{align}
We call such an instrument regular if $M_x$ is invertible for every $x\in X$.

For later convenience, we first give an equivalent characterization of the information-erasure condition.

\begin{lem}[Reference-branch characterization of  information erasure]
    \label{lem:reference_branch_condition}

    A single-Kraus instrument $\{M_x\}_{x\in X}$ is an information-erasure-inducing instrument if and only if
    there exist $P_{y^\ast}\in\operatorname{Adm}(\mathcal{I})$ and, for some (equivalently, any) reference
    outcome $x^\ast\in X$, positive numbers $\{r_x\}_{x\in X}$ with $r_{x^\ast}=1$ such that
    \begin{align}
        M_x^\dagger P_{y^\ast}M_x = r_x M_{x^\ast}^\dagger P_{y^\ast}M_{x^\ast},\qquad\forall x\in X.
        \label{eq:reference_branch_condition}
    \end{align}
    In this case, the corresponding posterior distribution is
    \begin{align}
        p(x|y^\ast) = \frac{r_x}{\sum_{z\in X}r_z}.
        \label{eq:q_from_reference_ratio}
    \end{align}
\end{lem}
\noindent
The probabilistic meaning of this lemma is that, if the information erasure is realized, the joint probability
distribution is proportional to the reference branch $x^\ast$ for every branch $x\in X$:
\begin{align}
    p(x,y^\ast|\rho) = r_x p(x^\ast,y^\ast|\rho),\qquad\forall x\in X.
\end{align}

The proof is given in Appendix~\ref{app:structural_proofs}.
Lemma~\ref{lem:reference_branch_condition} will be used throughout the following analysis.

Although the reference-branch characterization provides a simple operator condition for information erasure,
it is still difficult to solve the condition for the quantum instrument in general.
However, as shown in the following theorems, we can identify classes of quantum instruments that exhibit
information erasure or quantum imprint, and show that the classification is nontrivial.

\begin{prop}[Regular instruments under rank-one postselection]
    \label{prop:regular_erasure}

    Let $\{M_x\}_{x\in X}$ be a regular quantum instrument.
    Then there exists a rank-one POVM element $P_{y^\ast}$ that realizes information erasure if and only if,
    for some (and hence every) $x^\ast\in X$, the operators
    \begin{align}
        \left\{\left(M_xM_{x^\ast}^{-1}\right)^\dagger\right\}_{x\in X}
    \end{align}
    have a common eigenvector.

\end{prop}
\noindent
The proof is given in Appendix~\ref{app:regular_instruments}.
Equation~\eqref{eq:reference_branch_condition} can be rewritten as
\begin{align}
    \left(M_xM_{x^\ast}^{-1}\right)^\dagger P_{y^\ast} \left(M_xM_{x^\ast}^{-1}\right) = r_x P_{y^\ast}.
\end{align}
Intuitively, suppose that $\left(M_xM_{x^\ast}^{-1}\right)^\dagger$ have a normalized common eigenvector
$\ket{v}$ with corresponding eigenvalues $\lambda_x$,
\begin{align}
    \left(M_xM_{x^\ast}^{-1}\right)^\dagger \ketbra{v}\left(M_xM_{x^\ast}^{-1}\right) = |\lambda_x|^2\ketbra{v}
\end{align}
holds for every $x\in X$.
Therefore by considering $\ketbra{v}$ as $P_{y^\ast}$ and $|\lambda_x|^2\in \mathbb{R}$ as $r_x$, information
erasure can be realized.

\begin{cor}[Binary regular instruments]
    \label{cor:binary_regular_erasure}

    Let $X=\{0,1\}$ and suppose that the quantum instrument $(M_0,M_1)$ is regular.
    Then there always exists a POVM element $P_{y^\ast}$ that realizes information erasure.
\end{cor}
\begin{proof}
    Choose $x^\ast=0$.
    Since $X=\{0,1\}$, it is sufficient to consider the single operator
    \begin{align}
        \left(M_1M_0^{-1}\right)^\dagger .
    \end{align}
    In finite dimensions, this operator always has an eigenvector.
    Hence, by Proposition~\ref{prop:regular_erasure}, there exists a rank-one POVM element $P_{y^\ast}$ that
    realizes information erasure.
\end{proof}

So far, we have characterized information erasure for a given instrument.
We now consider the complementary problem in which the postselection effect is fixed.
\begin{prop}[Regular instruments for a full-rank postselection]
    \label{prop:full_rank_postselection_characterization}

    Fix an invertible POVM effect $P_{y^\ast}$. A regular single-Kraus instrument realizes information erasure
    under $P_{y^\ast}$ if and only if its Kraus operators admit the representation
    \begin{align}
        M_x = \sqrt{r_x}\,P_{y^\ast}^{-1/2}U_x\left(\sum_{z\in X}r_z U_z^\dagger P_{y^\ast}^{-1}U_z\right)^{-1/2},
        \label{eq:full_rank_postselection_family}
    \end{align}
    where $r_x>0$ is arbitrary and $U_x$ is arbitrary unitary operator for every $x\in X$.
\end{prop}
\noindent
The proof is given in Appendix~\ref{app:full_rank_postselection}.
This gives the complete family of regular instruments associated with a full-rank postselection.
Furthermore, in the qubit case, every nonzero admissible postselection effect is either rank one or full rank.
Therefore, Propositions~\ref{prop:regular_erasure} and~\ref{prop:full_rank_postselection_characterization}
together provide a complete characterization of information erasure for regular single-Kraus qubit
instruments.

Importantly, information erasure within this family is compatible with informational completeness when the
instrument is expressed by an associated POVM.
Recall that a POVM $\{E_x\}_{x\in X}$ is informationally complete if
$\operatorname{span}_{\mathbb{R}}\{E_x\}_{x\in X} =\operatorname{Herm}(\mathcal{H})$, equivalently, if its
outcome probabilities uniquely determine an arbitrary input state~\cite{Prugovecki1977,DArianoPerinottiSacchi2004}.

\begin{thm}[Informational completeness under information erasure]

    \label{thm:informationally_complete_erasure}

    Let $d:=\dim(\mathcal{H})\geq 2$ and $n:=|X|$.
    Fix an invertible POVM effect $P_{y^\ast}$ and arbitrary positive numbers $\{r_x\}_{x\in X}$, and consider
    the family of regular information-erasure-inducing instruments in
    Eq.~\eqref{eq:full_rank_postselection_family}.

    There exists a choice of unitaries $\{U_x\}_{x\in X}$ such that the associated POVM
    $\{E_x=M_x^\dagger M_x\}_{x\in X}$ is informationally complete if and only if
    \begin{align}
        P_{y^\ast}\not\propto I \qquad\text{and}\qquad n\geq d^2 .
    \end{align}
\end{thm}
\noindent
The proof is given in Appendix~\ref{app:informationally_complete_erasure}.

Thus, information erasure is not restricted to first measurements with limited state-discrimination
capability. A measurement can be informationally complete before postselection, and hence sufficient in
principle for full state tomography, while the distribution of the same measurement record becomes completely
independent of the input state after an appropriate postselection.
This shows that information erasure cannot be attributed simply to a lack of information acquired by the first
measurement.

Having discussed information-erasure-inducing instruments, we now turn to quantum imprint.
We begin with projective measurements, which provide a familiar and fundamental example.
\begin{prop}[Projective measurements are quantum imprints]
    \label{prop:projective_quantum_imprint}
    Let $\{M_x\}_{x\in X}$ be a single-Kraus instrument with $|X|\geq 2$, and suppose that
    $\{M_x^\dagger M_x\}_{x\in X}$ is a PVM.
    Then no admissible postselection effect $P_{y^\ast}$ can realize information erasure.
    Hence, the instrument is a quantum imprint.
\end{prop}
\noindent
The proof is given in Appendix~\ref{app:projective_imprint}.
Indeed, each operator $M_x^\dagger P_{y^\ast}M_x$ has support contained in
$\operatorname{Ran}(M_x^\dagger M_x)$.
Since the PVM elements corresponding to different branches have orthogonal ranges, these operators have
mutually orthogonal supports.
On the other hand, from Lemma~\ref{lem:reference_branch_condition}, information erasure would require
\begin{align}
    M_x^\dagger P_{y^\ast}M_x = r_x M_{x^\ast}^\dagger P_{y^\ast}M_{x^\ast}, \qquad r_x>0,
\end{align}
for all $x$, so that all branches have the same nonzero support.
This is impossible.

Quantum imprint is not restricted to projective measurements.
The following proposition shows that it can also arise within regular instruments.

\begin{prop}[Regular quantum imprints]
    \label{prop:regular_imprint_arbitrary_dimension}
    Let $d:=\dim(\mathcal{H})\geq 2$.
    Then there exists a three-outcome regular single-Kraus instrument that is a quantum imprint.
    Let
    \begin{align}
        S_d &:= \sum_{j = 1}^{d-1}\ketbra{j}{j+1}, \notag \\
        A_0 &:= I,\ A_1 := I-aS_d,\ A_2 := I+b S_d^\dagger,
    \end{align}
    where $a,b>0$, and define
    \begin{align}
        M_x := A_x\left(\sum_{z = 0}^{2}A_z^\dagger A_z\right)^{-1/2},\quad x\in\{0,1,2\}.
    \end{align}
    Then $\{M_x\}_{x=0}^2$ is regular and a quantum imprint.
\end{prop}
\noindent
The proof is given in Appendix~\ref{app:regular_imprint_arbitrary_dimension}.

Together with Proposition~\ref{prop:regular_erasure}, this result shows that both information erasure and
quantum imprint can occur within regular instruments.
Thus, the distinction between the two cannot be determined solely by whether the instrument is regular.
This naturally raises the question of which structural features of the instrument are responsible for the
distinction.
Since a quantum instrument contains more information than its associated POVM
$\{E_x=M_x^\dagger M_x\}_{x\in X}$, we next ask whether the distinction is already visible at the POVM level.

The following theorem shows that it is not.
\begin{thm}[POVM-equivalent instruments with opposite information structures]
    \label{thm:same_povm_erasure_imprint}

    For every $d\geq 2$, there exists a three-outcome POVM $\{E_x\}_{x=0}^2$ admitting two regular
    single-Kraus realizations $\{M_x^{\mathrm{er}}\}_{x=0}^2$ and $\{M_x^{\mathrm{imp}}\}_{x=0}^2$ such that
    \begin{align}
        (M_x^{\mathrm{er}})^\dagger M_x^{\mathrm{er}} = (M_x^{\mathrm{imp}})^\dagger M_x^{\mathrm{imp}} = E_x
    \end{align}
    for every $x$, while the former is information-erasure-inducing and the latter is a quantum imprint.

\end{thm}
\noindent
The proof is given in Appendix~\ref{app:same_povm_erasure_imprint}.
The idea is simple.
Proposition~\ref{prop:regular_imprint_arbitrary_dimension} provides a regular quantum imprint with POVM
$\{E_x\}_{x\in X}$.
For a fixed POVM, however, the Kraus operators are not uniquely determined: their unitary parts can be changed
without affecting $M_x^\dagger M_x=E_x$.
These unitary degrees of freedom leave all first measurement outcome probabilities unchanged while modifying
the postmeasurement states.
They can therefore be chosen so that the common-eigenvector condition in
Proposition~\ref{prop:regular_erasure} is satisfied, yielding an information-erasure-inducing realization of
the same POVM.

At first sight, one might expect the distinction between information erasure and quantum imprint to be closely
related to how strongly the first measurement acquires information about the input state.
However, the results above show that this intuition is insufficient.
Information erasure and quantum imprint can both occur within regular instruments, and even quantum
instruments implementing exactly the same POVM can exhibit opposite behaviors under postselection, as shown in
Theorem~\ref{thm:same_povm_erasure_imprint}.

Thus, the distinction cannot be reduced to a simple difference between projective and regular measurements,
nor can it be determined solely from the informativeness of the associated POVM.
In particular, information erasure and quantum imprint reflect a structure of the full quantum instrument,
including the postmeasurement state transformation, that is not captured by the POVM alone.

Importantly, this distinction is not merely taxonomic.
In the following sections, we show that the information-erasure structure can be directly exploited for
concrete operational tasks.

\section{Applications to SPAM noise separation}
In the previous section, we established the nontrivial structure of information erasure and quantum imprint.
We now turn to an operational consequence of information erasure in the presence of small experimental
imperfections.
In particular, we investigate how perturbations in state preparation, the measurement instrument, and
postselection appear in the posterior distribution around an ideal information-erasure point.

This problem is closely related to the characterization of state-preparation-and-measurement (SPAM) errors.
In realistic quantum devices, imperfections in state preparation and measurement generally occur
simultaneously, so that their contributions are mixed in the observed statistics.
Consequently, the observed statistics alone do not generally reveal whether a deviation originates from state
preparation or from measurement, making their separate characterization intrinsically difficult.
Various approaches have been proposed to address this ambiguity by introducing additional information or
structural assumptions, including self-consistent characterization and related methods~\cite{Nielsen2021GST,Cattaneo2023SelfConsistentQMT,Jayakumar2024Universal,Lin2021IndependentSPAM,Yu2025SeparateSPAM,Stark2014,JacksonVanEnk2015}.

Here, we show that information erasure provides a different route to this problem.
If a postselection realizing information erasure can be implemented with sufficient reliability, in the sense
that its imperfections do not affect the posterior distribution at first order, the first-order contribution
of state-preparation error to the posterior distribution vanishes, whereas visible perturbations of the
measurement instrument remain.
Thus, the problem of separating state-preparation and measurement errors can be reformulated as the problem of
realizing a reliable postselection under which the preparation-side contribution is suppressed.
We first establish this separation by analyzing perturbations around an ideal information-erasure point.

\subsection{Erasure-Induced First-Order Noise Separation}
\label{sec:first_order_separation}

We focus on the single-Kraus instruments considered in Sec.~\ref{sec:structural_results}.
Let $\boldsymbol{M}=\{M_x\}_{x\in X}$ be a single-Kraus instrument, and let $P_{y^\ast}$ be an admissible
postselection effect such that $(\boldsymbol{M},P_{y^\ast})$ realizes information erasure.
Define $F_x:=M_x^\dagger P_{y^\ast}M_x$ and $F:=\sum_zF_z$.
By Definition~\ref{dfn:information_erasure_inducing_instrument}, there exists a full-support probability
distribution $q=(q_x)_{x\in X}$ satisfying
\begin{align}
    F_x = q_xF,\qquad \forall x\in X.
    \label{eq:ideal_erasure_spam}
\end{align}

Now, we introduce a common perturbative scale $\epsilon$ and model the prepared state, measurement instrument,
and postselection effect as
\begin{align}
    \widetilde{\rho} &= \rho+\delta\rho+O(\epsilon^2), \notag \\
    \widetilde{M}_x &= M_x+\delta M_x+O(\epsilon^2), \notag \\
    \widetilde{P}_{y^\ast} &= P_{y^\ast}+\delta P_{y^\ast}+O(\epsilon^2).
    \label{eq:spam_noise_model}
\end{align}
Here, $\delta\rho$, $\delta M_x$, and $\delta P_{y^\ast}$ are first-order perturbations of order
$O(\epsilon)$.
The measurement perturbation is assumed to preserve the single-Kraus structure of each outcome branch.
Moreover, since the perturbed operators must continue to define a quantum instrument, the completeness
relation requires, to first order,
\begin{align}
    \sum_{x\in X} \left(\delta M_x^\dagger M_x + M_x^\dagger\delta M_x\right) = 0.
    \label{eq:first_order_completeness}
\end{align}

For the instrument perturbation, define
\begin{align}
    \delta F_x &:= \delta M_x^\dagger P_{y^\ast}M_x + M_x^\dagger P_{y^\ast}\delta M_x, \notag \\
    \delta F &:= \sum_z\delta F_z,\qquad R_x := \delta F_x-q_x\delta F.
    \label{eq:instrument_residual_operator}
\end{align}
Similarly, for the postselection perturbation, define
\begin{align}
    \delta F_x^{P} &:= M_x^\dagger\delta P_{y^\ast}M_x, \notag \\
    \delta F^{P} &:= \sum_z\delta F_z^{P},\qquad R_x^{P} := \delta F_x^{P}-q_x\delta F^{P}.
    \label{eq:postselection_residual_operator}
\end{align}
With these definitions, we obtain the following result of first-order noise-separation.

\begin{thm}[First-order SPAM error separation]
    \label{thm:first_order_spam_separation}
    Under the noise model in Eq.~\eqref{eq:spam_noise_model}, suppose that the postselection perturbation
    satisfies
    \begin{align}
        R_x^{P} = 0,\qquad \forall x\in X.
        \label{eq:postselection_kernel_condition}
    \end{align}
    Then the implemented posterior distribution satisfies
    \begin{align}
        \widetilde{p}\bigl(x|y^\ast,\widetilde{\rho}\bigr) = q_x+ \frac{\operatorname{Tr}(\rho R_x)}{\operatorname{Tr}(\rho F)} +O(\epsilon^2).
        \label{eq:first_order_separation}
    \end{align}
    Hence, the first-order deviation contains neither the unknown preparation perturbation $\delta\rho$ nor
    the postselection perturbation $\delta P_{y^\ast}$, and only the contribution of the instrument
    perturbation encoded in $R_x$ remains.
\end{thm}
The cancellation of the state-preparation error follows directly from the information-erasure condition
$F_x=q_xF$, since
\begin{align}
    \operatorname{Tr}\left[\delta\rho\left(F_x-q_xF\right)\right] = 0.
\end{align}
By contrast, the instrument perturbation does not generally cancel and remains through $R_x$. A complete
derivation is given in Appendix~\ref{app:first_order_spam_separation}.

The simplest case when $R^P_x=0$ is $\delta P_{y^\ast}=0$.
Thus, if the postselection effect realizing information erasure can be implemented reliably,
Theorem~\ref{thm:first_order_spam_separation} shows that state-preparation and measurement errors can be
separated to first order without precise knowledge of the actually prepared state.
Importantly, however, precise implementation of the postselection effect is not necessary.
Nonzero perturbations $\delta P_{y^\ast}$ may also satisfy Eq.~\eqref{eq:postselection_kernel_condition}.
The set of such perturbations is the kernel of the real-linear map
\begin{align}
    T_P:\delta P_{y^\ast} \longmapsto \{R_x^P\}_{x\in X}.
    \label{eq:postselection_noise_map}
\end{align}
This kernel can contain error directions beyond a uniform rescaling of the postselection effect.
In Appendix~\ref{app:minimal_spam_sim}, we construct a qubit example and show that a small unitary rotation
about the $y$ axis is invisible to first order.

For this first-order separation mechanism to have operational significance, three issues need to be addressed:
\begin{enumerate}
    \item Is the protocol applicable to a sufficiently broad class of measurements, despite requiring an
    information-erasure-inducing instrument?
    \item What information about measurement errors can be extracted from the posterior distribution?
    \item Can the observable signal be used to reduce the error of the measurement apparatus itself?
\end{enumerate}

The first issue is addressed in part by the structural results of Sec.~\ref{sec:structural_results}.
As shown in Theorem~\ref{thm:informationally_complete_erasure}, an information-erasure-inducing instrument can
implement an informationally complete POVM.
Furthermore, Theorem~\ref{thm:same_povm_erasure_imprint} shows that, for some POVMs, distinct
quantum-instrument realizations of exactly the same POVM can exhibit opposite information structures, with one
admitting information erasure and the other exhibiting quantum imprint.
Thus, requiring information erasure does not necessarily restrict the first measurement to POVMs with limited
state-discrimination capability, nor does it necessarily require changing the first-measurement outcome
statistics.
In particular, when the information-erasure-inducing instrument implements an informationally complete POVM,
the same instrument can, in principle, be used without postselection for full state tomography after
calibration.

In the next two subsections, we show that the posterior distribution determines the residual operators $R_x$,
which capture the components of the measurement error visible to the present protocol.
We then show how these residuals can be used for feedback calibration and derive a sufficient condition under
which reducing the observable residual decreases the first-order measurement-error norm.
As a consequence, we obtain a tighter state-uniform bound on the deviation of the implemented outcome
probabilities from the ideal ones.
\subsection{Observable Measurement Errors}
Theorem~\ref{thm:first_order_spam_separation} shows that the first-order deviation of the posterior
distribution is determined solely by the measurement residual $R_x$.
Rearranging Eq.~\eqref{eq:first_order_separation} gives
\begin{align}
    \operatorname{Tr}\!\left[\rho R_x\right] = \operatorname{Tr}(\rho F)\left[\widetilde{p}(x| y^\ast,\widetilde{\rho})-q_x\right]+O(\epsilon^2).
    \label{eq:Rx_reconstruction}
\end{align}
The quantities $\operatorname{Tr}(\rho F)$ and $q_x$ are determined by the ideal information-erasure pair and
the ideal quantum state $\rho$, whereas $\widetilde{p}(x| y^\ast,\widetilde{\rho})$ can be estimated
experimentally.
Therefore, by repeating the experiment with a sufficiently rich set of input states, each $R_x$ can be
reconstructed through a tomography-like procedure, provided that the second-order terms $O(\epsilon^2)$ are
negligible.
Since $R_x$ is defined in Eq.~\eqref{eq:instrument_residual_operator} in terms of the instrument perturbations
$\{\delta M_x\}_{x\in X}$, this reconstruction answers the second question posed above.

It is important, however, that the map from the underlying measurement perturbation
$\delta M=\{\delta M_x\}_{x\in X}$ to the experimentally accessible residuals $\{R_x\}_{x\in X}$ is not
necessarily injective.
Hence, reconstructing every $R_x$ does not in general amount to reconstructing the full measurement error
$\delta M$.

To make this distinction explicit, let
\begin{align}
    V &:= \left\{\delta M = \{\delta M_x\}_{x\in X} \;\middle|\; \sum_{x\in X} \left(\delta M_x^\dagger M_x + M_x^\dagger \delta M_x\right) = 0\right\}\notag, \\
    W &:= \bigoplus_{x\in X}\operatorname{Herm}(\mathcal{H}),
\end{align}
where both spaces are regarded as real vector spaces, and define the real-linear map
\begin{align}
    T:V\longrightarrow W, \qquad T(\delta M) := \{R_x\}_{x\in X}.
    \label{eq:T_map}
\end{align}
We equip $V$ and $W$ with the real Hilbert--Schmidt inner products
\begin{align}
    \langle \{A_x\}_x,\{B_x\}_x\rangle_V &:= \sum_x \operatorname{Re}\operatorname{Tr}(A_x^\dagger B_x), \\
    \langle \{C_x\}_x,\{D_x\}_x\rangle_W &:= \sum_x \operatorname{Tr}(C_xD_x).
\end{align}
The corresponding norms are induced by these inner products.

We then define
\begin{align}
    \mathcal{K} := \ker T, \qquad V_{\mathrm{vis}} := \mathcal{K}^\perp.
\end{align}
The kernel $\mathcal{K}$ consists of measurement perturbations that do not contribute to the residuals $R_x$
and are therefore invisible to the present protocol.

A simple and physically important example of such a kernel component is a common right-unitary error.
Consider
\begin{align}
    \widetilde{M}_x = M_xV(\epsilon), \qquad V(\epsilon) = e^{-i\epsilon H}, \qquad H = H^\dagger.
\end{align}
To first order,
\begin{align}
    \delta M_x = -i\epsilon M_xH,
\end{align}
and hence
\begin{align}
    \delta F_x = i\epsilon[H,F_x].
\end{align}
Using the ideal information-erasure condition $F_x=q_xF$, we obtain
\begin{align}
    \delta F_x = q_x i\epsilon[H,F], \qquad \delta F = i\epsilon[H,F],
\end{align}
and therefore
\begin{align}
    R_x = \delta F_x-q_x\delta F = 0
\end{align}
for every $x\in X$.
Thus, every common right-unitary perturbation belongs to $\ker T$.

A common unitary acting from the right of all measurement operators can equivalently be regarded as a unitary
acting on the input state before the measurement:
\begin{align}
    (M_xV)\rho(M_xV)^\dagger = M_x(V\rho V^\dagger)M_x^\dagger.
    \label{eq:right_unitary_preparation_equivalence}
\end{align}
Therefore, from the observed statistics alone, such a measurement error cannot be distinguished from a
coherent preparation error.
Since the present protocol is designed to eliminate preparation-side perturbations at first order, it is
necessarily insensitive to measurement-error components that are operationally equivalent to them.

This visible component is the part that can subsequently be used as a feedback signal for calibrating the
measurement apparatus.
\subsection{Feedback Calibration}
\label{sec:guaranteed_feedback}
We finally address the third question posed above: whether the observable residuals $R_x$ can be used as a
feedback signal to reduce the error of the measurement apparatus itself.

Let $\mathcal{F}_{\eta}$ be a smooth feedback map parametrized by $\eta\in\mathbb{R}^m$, with
$\mathcal{F}_0=\mathrm{id}$, and let $\mathcal{N}$ denote an unknown measurement-noise map.
We restrict ourselves to the perturbative regime in which both the measurement noise and the feedback
parameter $\eta$ are sufficiently small.
In what follows, $O(2)$ collectively denotes terms of second order or higher in these small quantities,
including $O(\epsilon^2)$, $O(\epsilon\|\eta\|)$, and $O(\|\eta\|^2)$.

The implemented instrument is written as
\begin{align}
    \widetilde{\boldsymbol{M}}(\eta) &:= \mathcal{N}\circ\mathcal{F}_\eta(\boldsymbol{M}) \notag \\
    &= \{M_x+\delta M_x(\eta)+O(2)\}_{x\in X},
    \label{eq:effective_feedback_perturbation}
\end{align}
where $\delta M(\eta)=\{\delta M_x(\eta)\}_{x\in X}\in V$ is the total first-order measurement perturbation.
Its observable component is characterized by
\begin{align}
    T(\delta M(\eta)) = \{R_x(\eta)\}_{x\in X}.
\end{align}
Therefore, the natural feedback cost is defined as
\begin{align}
    C(\eta) &:= \sum_{x\in X}\|R_x(\eta)\|_{\mathrm{HS}}^2 = \|T(\delta M(\eta))\|_W^2.
    \label{eq:feedback_cost}
\end{align}
However, reducing this cost does not necessarily imply a reduction in the norm of the underlying measurement
perturbation $\delta M(\eta)$.

Here, we introduce the quantities required to derive a sufficient condition under which a reduction of the
observable cost guarantees a reduction of the first-order measurement-error norm.
Since only the first-order action of the feedback is relevant in the perturbative regime, we expand the
feedback map around $\eta=0$ as
\begin{align}
    \mathcal{F}_\eta(\boldsymbol{M}) = \boldsymbol{M} + \sum_{j = 1}^{m}\eta_j B_j + O(\|\eta\|^2),
\end{align}
where
\begin{align}
    B_j := \left. \frac{\partial \mathcal{F}_\eta(\boldsymbol{M})}{\partial \eta_j}\right|_{\eta = 0} \in V.
\end{align}
Because $\mathcal{F}_\eta(\boldsymbol{M})$ is a quantum instrument for sufficiently small $\eta$, each $B_j$
satisfies the first-order completeness condition and hence belongs to $V$.
In this work, we restrict our attention to feedback directions that are visible through $T$, and therefore
assume
\begin{align}
    B_j\in V_{\mathrm{vis}}, \qquad j = 1,\ldots,m.
\end{align}
We then define the feedback-accessible subspace by
\begin{align}
    V_{\mathrm{fb}} := \operatorname{span}_{\mathbb{R}} \{B_1,\ldots,B_m\} \subseteq V_{\mathrm{vis}}.
    \label{eq:Vfb_definition}
\end{align}
We further define the noncontrollable visible subspace by
\begin{align}
    V_{\mathrm{nc}} := \left\{v\in V_{\mathrm{vis}} \,\middle|\, \langle T(v),T(w)\rangle_W = 0 \text{ for all }w\in V_{\mathrm{fb}}\right\}.
    \label{eq:Vnc_definition}
\end{align}

To quantify the response of the observable residuals along the controllable directions, define
\begin{align}
    \alpha &:= \min_{v\in V_{\mathrm{fb}}\setminus\{0\}} \frac{\|T(v)\|_W}{\|v\|_V}, \qquad \beta := \max_{v\in V_{\mathrm{fb}}\setminus\{0\}} \frac{\|T(v)\|_W}{\|v\|_V}.
    \label{eq:alpha_beta_feedback}
\end{align}

Importantly, $T$ is determined by the ideal information-erasure pair, while $V_{\mathrm{fb}}$ is determined by
the specified feedback map around the ideal instrument.
Therefore, $V_{\mathrm{nc}}$, $\alpha$, $\beta$, and the geometric conditions appearing below can, in
principle, be determined in advance without knowledge of the unknown noise $\mathcal{N}$.

With these quantities, we can now state a sufficient condition under which reducing the experimentally
accessible feedback cost guarantees a reduction of the underlying first-order measurement-error norm.
\begin{thm}[Feedback guarantee theorem]
    \label{thm:feedback_guarantee}

    Suppose that the feedback changes only the $V_{\mathrm{fb}}$ component of the first-order measurement
    perturbation.
    In addition, assume that either
    \begin{enumerate}
        \item $V_{\mathrm{nc}}=0$, or
        \item $V_{\mathrm{nc}}\perp V_{\mathrm{fb}}$ with respect to the inner product on $V$.
    \end{enumerate}

    If the final feedback parameter $\eta_{\mathrm{final}}$ satisfies
    \begin{align}
        C(\eta_{\mathrm{final}}) < \left(\frac{\alpha}{\beta}\right)^2 C(0),
    \end{align}
    then,
    \begin{align}
        \left\| \mathcal{N}\circ\mathcal{F}_{\eta_{\mathrm{final}}}(\boldsymbol{M}) - \boldsymbol{M}\right\|_V < \left\| \mathcal{N}(\boldsymbol{M}) - \boldsymbol{M}\right\|_V +O(2).
    \end{align}
    Thus, under the assumptions of the theorem, the feedback is guaranteed to reduce the Hilbert--Schmidt norm
    of the first-order Kraus-operator perturbation.
\end{thm}
\noindent
The proof is given in Appendix~\ref{app:feedback_control}.

In particular, $V_{\mathrm{nc}}=0$ means that every error direction visible to the present protocol is
controllable by the available feedback.
Hence, if the feedback spans the entire visible error space, $V_{\mathrm{fb}}=V_{\mathrm{vis}}$, the first
condition of Theorem~\ref{thm:feedback_guarantee} is automatically satisfied.

The perturbation norm appearing in Theorem~\ref{thm:feedback_guarantee} also has a direct consequence for the
measurement statistics.
Specifically, a reduction of this norm tightens a state-uniform upper bound on the deviation of the
implemented outcome probabilities from the ideal ones, as shown in the following corollary.
\begin{cor}[Uniform improvement of the measurement-statistics error bound]
    \label{thm:uniform_measurement_improvement}

    Suppose that the feedback satisfies the assumptions of Theorem~\ref{thm:feedback_guarantee} and let
    $\eta_{\mathrm{final}}$ satisfy $C(\eta_{\mathrm{final}})<(\alpha/\beta)^2C(0)$.

    For an arbitrary single-Kraus instrument $\mathcal{A}=\{A_x\}_{x\in X}$, define
    $p_{\mathcal{A}}(x|\rho):=\operatorname{Tr}(\rho A_x^\dagger A_x)$.

    Then,
    \begin{align}
        & \sup_{\rho\in S(\mathcal{H})} \frac12\sum_{x\in X} \left| p_{\mathcal{N}\circ\mathcal{F}_{\eta_{\mathrm{final}}}(\boldsymbol{M})}(x|\rho) - p_{\boldsymbol{M}}(x|\rho)\right| \notag \\
        &\quad\leq \sqrt{|X|} \left\| \mathcal{N}\circ\mathcal{F}_{\eta_{\mathrm{final}}}(\boldsymbol{M}) - \boldsymbol{M}\right\|_V +O(2), \notag \\
        & \sup_{\rho\in S(\mathcal{H})} \frac12\sum_{x\in X} \left| p_{\mathcal{N}(\boldsymbol{M})}(x|\rho) - p_{\boldsymbol{M}}(x|\rho)\right| \notag \\
        &\quad\leq \sqrt{|X|} \left\| \mathcal{N}(\boldsymbol{M}) - \boldsymbol{M}\right\|_V +O(2).
        \label{eq:uniform_probability_bound_improvement}
    \end{align}

    Furthermore, Theorem~\ref{thm:feedback_guarantee} implies
    \begin{align}
        & \sqrt{|X|} \left\| \mathcal{N}\circ\mathcal{F}_{\eta_{\mathrm{final}}}(\boldsymbol{M}) - \boldsymbol{M}\right\|_V \notag \\
        &\quad < \sqrt{|X|} \left\| \mathcal{N}(\boldsymbol{M}) - \boldsymbol{M}\right\|_V +O(2).
    \end{align}
    Therefore, after feedback, the state-uniform upper bound on the deviation of the implemented measurement
    statistics from the ideal statistics is strictly smaller than the corresponding bound before feedback.
\end{cor}
\noindent
The proof is given in Appendix~\ref{app:uniform_measurement_improvement}.

Theorem~\ref{thm:feedback_guarantee} and Corollary~\ref{thm:uniform_measurement_improvement} answer the third
question posed above.
Taken together, these results show that information erasure has direct operational significance.

\section{Summary and Discussion}
In this work, we introduced information erasure and quantum imprint as two properties that classify quantum
instruments according to whether postselection can eliminate the dependence of an earlier measurement record
on the initial quantum state while retaining all outcome branches.
We showed that this classification has a nontrivial and counterintuitive structure.
In particular, information erasure is compatible with informational completeness, so that even an instrument
whose associated POVM can uniquely identify the input state may exhibit information erasure.
Moreover, two quantum instruments implementing exactly the same POVM, and therefore producing identical
outcome statistics before postselection, can exhibit opposite behaviors, with one admitting information
erasure and the other exhibiting quantum imprint.
These results show that the natural intuition that measurements providing more information about the initial
quantum state should be less likely to exhibit information erasure does not hold in general.

We further showed that information erasure has direct operational significance for first-order SPAM error
separation.
If a postselection realizing information erasure can be implemented with sufficient reliability, in the sense
that its imperfections do not affect the posterior distribution at first order, the first-order contribution
of state-preparation error vanishes, whereas visible first-order contributions of measurement error remain.
Thus, first-order SPAM error separation can be reformulated from the simultaneous characterization of state
preparation and measurement into the problem of realizing a sufficiently reliable postselection.

Taken together, these results show that the distinction between information erasure and quantum imprint is not
merely a formal classification of quantum instruments, but reflects a structure that is relevant both to
fundamental questions in quantum measurement and to the development of quantum technologies.
This perspective raises several open questions from both foundational and applied viewpoints.

From a foundational perspective, an important problem is to obtain a complete characterization of information
erasure and quantum imprint for more general quantum instruments.
Once such a classification is established, it will be important to clarify how it is related to existing
notions that quantify the information acquired by a measurement or its information-gathering capability.
Our results indicate that information erasure and quantum imprint reflect a structure that cannot be captured
solely by the information content of the associated POVM.
Understanding the relation between this structure and existing information-theoretic quantities is therefore
an important open problem.

From an applied perspective, it will be important to experimentally test the first-order SPAM-error-separation
and feedback-calibration schemes proposed in this work.
It is also necessary to systematically characterize the space of postselection perturbations satisfying
$R_x^P=0$, and to understand how its dimension and structure determine the robustness of the protocol against
experimental imperfections.
Furthermore, extending both the structural results and the SPAM-error-separation framework beyond the
single-Kraus setting to general quantum instruments with multiple Kraus operators for each outcome is another
important direction.

Clarifying these questions may allow information erasure and quantum imprint to develop beyond SPAM-error
separation into a more general framework for understanding how information is retained and erased in
sequential quantum measurements, and for exploiting this structure in the design and control of quantum
measurements.

\section*{Acknowledgements}
This study was supported by the Cross-ministerial Strategic Innovation Promotion Program (SIP) of the Cabinet
Office (No. 23836436). YI is supported by MEXT Quantum Leap Flagship Program (MEXT Q-LEAP) Grant No.
JPMXS0120319794, JST COI-NEXT Grant No. JPMJPF2014, JST CREST JPMJCR24I3, JST SPRING Grant No. JPMJSP2138 and
the $\Sigma$ Doctoral Futures Research Grant Program from The University of Osaka.
TS and YI thank the Yukawa Institute for Theoretical Physics at Kyoto University, where this work was initiated during the YITP-W-25-11.

\bibliography{references}

\appendix
\section{Proofs of the Structural Results}
\label{app:structural_proofs}

In this appendix, we provide the proofs of the structural results presented in
Sec.~\ref{sec:structural_results}.
\subsection{Proof of Lemma~\ref{lem:reference_branch_condition}}
\label{app:reference_branch}

We prove the equivalence between the original information-erasure condition and the reference-branch
representation.

\begin{proof}

    Suppose first that the instrument is information-erasure-inducing.
    Then there exist $P_{y^\ast}\in\operatorname{Adm}(\mathcal{I})$ and $q\in\Delta_X^\circ$ such that
    \begin{align}
        M_x^\dagger P_{y^\ast}M_x = q_x\sum_{z\in X}M_z^\dagger P_{y^\ast}M_z,\qquad\forall x\in X.
        \label{eq:appendix_original_erasure}
    \end{align}
    Fix any reference outcome $x^\ast\in X$.
    Since $q_{x^\ast}>0$, Eq.~\eqref{eq:appendix_original_erasure} for $x=x^\ast$ gives
    \begin{align}
        M_{x^\ast}^\dagger P_{y^\ast}M_{x^\ast} = q_{x^\ast}\sum_{z\in X}M_z^\dagger P_{y^\ast}M_z.
    \end{align}
    Therefore,
    \begin{align}
        M_x^\dagger P_{y^\ast}M_x = \frac{q_x}{q_{x^\ast}}M_{x^\ast}^\dagger P_{y^\ast}M_{x^\ast},\qquad\forall x\in X.
    \end{align}
    Defining $r_x:=q_x/q_{x^\ast}$, we have $r_x>0$ and $r_{x^\ast}=1$, and hence
    Eq.~\eqref{eq:reference_branch_condition} follows.
    Conversely, suppose that Eq.~\eqref{eq:reference_branch_condition} holds for some $x^\ast\in X$,
    $P_{y^\ast}\in\operatorname{Adm}(\mathcal{I})$, and positive numbers $\{r_x\}_{x\in X}$ with
    $r_{x^\ast}=1$.
    Summing over $x$ gives
    \begin{align}
        \sum_{z\in X}M_z^\dagger P_{y^\ast}M_z = \left(\sum_{z\in X}r_z\right)M_{x^\ast}^\dagger P_{y^\ast}M_{x^\ast}.
    \end{align}
    It follows that
    \begin{align}
        M_x^\dagger P_{y^\ast}M_x = \frac{r_x}{\sum_{z\in X}r_z}\sum_{z\in X}M_z^\dagger P_{y^\ast}M_z,\qquad\forall x\in X.
    \end{align}
    Thus, defining
    \begin{align}
        q_x := \frac{r_x}{\sum_{z\in X}r_z},
    \end{align}
    we obtain $q\in\Delta_X^\circ$ and recover the information-erasure condition.
    The corresponding posterior distribution is therefore $p(x|y^\ast)=q_x$.

\end{proof}

\subsection{Proof of Proposition~\ref{prop:regular_erasure}}
\label{app:regular_instruments}
We now prove Proposition~\ref{prop:regular_erasure}.

\begin{proof}

    Fix an arbitrary reference outcome $x^\ast\in X$.

    Suppose first that a rank-one postselection effect $P_{y^\ast}$ realizes information erasure.
    Since $P_{y^\ast}$ is rank one, write
    \begin{align}
        P_{y^\ast} = c\ketbra{v}{v},\qquad0<c\leq1,\qquad\braket{v}{v} = 1.
    \end{align}
    By Lemma~\ref{lem:reference_branch_condition}, there exist $r_x>0$ such that
    \begin{align}
        M_x^\dagger\ketbra{v}{v}M_x = r_xM_{x^\ast}^\dagger\ketbra{v}{v}M_{x^\ast}\qquad\forall x\in X.
        \label{eq:appendix_rank_one_proportionality}
    \end{align}

    Set $\ket{u_x}:=M_x^\dagger\ket{v}$.
    Since every $M_x$ is regular, $\ket{u_x}\neq0$ for all $x$.
    Equation~\eqref{eq:appendix_rank_one_proportionality} gives
    \begin{align}
        \ketbra{u_x}{u_x} = r_x\ketbra{u_{x^\ast}}{u_{x^\ast}}.
    \end{align}
    The two rank-one operators therefore have the same range, so there exists
    $\lambda_x\in\mathbb{C}\setminus\{0\}$ such that
    \begin{align}
        \ket{u_x} = \lambda_x\ket{u_{x^\ast}}.
    \end{align}
    Substituting this relation into the preceding equality gives $|\lambda_x|^2=r_x$.
    Hence
    \begin{align}
        M_x^\dagger\ket{v} = \lambda_xM_{x^\ast}^\dagger\ket{v},
    \end{align}
    and therefore
    \begin{align}
        \left(M_xM_{x^\ast}^{-1}\right)^\dagger\ket{v} = \lambda_x\ket{v}\qquad\forall x\in X.
        \label{eq:appendix_common_eigenvector}
    \end{align}
    Thus $\{(M_xM_{x^\ast}^{-1})^\dagger\}_{x\in X}$ have the common eigenvector $\ket{v}$.

    Conversely, suppose that $\{(M_xM_{x^\ast}^{-1})^\dagger\}_{x\in X}$ have a common eigenvector
    $\ket{v}\neq0$.
    Let $\lambda_x$ denote the corresponding eigenvalues:
    \begin{align}
        \left(M_xM_{x^\ast}^{-1}\right)^\dagger\ket{v} = \lambda_x\ket{v}\qquad\forall x\in X.
    \end{align}
    Normalize $\ket{v}$ and define $P_{y^\ast}:=c\ketbra{v}{v},\ 0<c\leq1$.

    Multiplying the common-eigenvector relation by $M_{x^\ast}^\dagger$ gives
    \begin{align}
        M_x^\dagger\ket{v} = \lambda_xM_{x^\ast}^\dagger\ket{v}.
    \end{align}
    Taking the corresponding outer products yields
    \begin{align}
        M_x^\dagger P_{y^\ast}M_x = |\lambda_x|^2M_{x^\ast}^\dagger P_{y^\ast}M_{x^\ast}\qquad\forall x\in X.
    \end{align}
    Since every operator $\left(M_xM_{x^\ast}^{-1}\right)^\dagger$ is regular, its eigenvalue $\lambda_x$ is
    nonzero for every $x\in X$.
    Moreover, since $M_{x^\ast}$ is regular, $M_{x^\ast}^\dagger\ket{v}\neq0$, and hence
    \begin{align}
        M_{x^\ast}^\dagger P_{y^\ast}M_{x^\ast}\neq0.
    \end{align}
    Therefore $M_x^\dagger P_{y^\ast}M_x\neq 0$ for all $x\in X$, which means
    $P_{y^\ast}\in\operatorname{Adm}(\mathcal{I})$.

    Since $|\lambda_x|^2>0$, Lemma~\ref{lem:reference_branch_condition} implies that $P_{y^\ast}$ realizes
    information erasure.

    Because $x^\ast$ was chosen arbitrarily, the argument holds.
\end{proof}

\subsection{Proof of Proposition~\ref{prop:full_rank_postselection_characterization}}
\label{app:full_rank_postselection}
We next prove Proposition~\ref{prop:full_rank_postselection_characterization}.

\begin{proof}

    Let $P_{y^\ast}>0$ be fixed.

    We first prove the necessity.
    Suppose that a regular instrument $\{M_x\}_{x\in X}$ realizes information erasure under $P_{y^\ast}$.
    By Lemma~\ref{lem:reference_branch_condition}, there exist $r_x>0$ such that
    \begin{align}
        M_x^\dagger P_{y^\ast}M_x = r_xM_{x^\ast}^\dagger P_{y^\ast}M_{x^\ast}\qquad\forall x\in X.
        \label{eq:appendix_full_rank_reference}
    \end{align}
    Set $D=M_{x^\ast}^\dagger P_{y^\ast}M_{x^\ast}$.
    Since $P_{y^\ast}$ and $M_{x^\ast}$ are regular, $D>0$.

    Equation~\eqref{eq:appendix_full_rank_reference} is equivalent to
    \begin{align}
        \left(P_{y^\ast}^{1/2}M_x\right)^\dagger\left(P_{y^\ast}^{1/2}M_x\right) = r_xD.
    \end{align}
    Since $P_{y^\ast}^{1/2}M_x$ is regular, its polar decomposition has a unitary polar factor.
    Hence, for each $x\in X$, there exists a unitary $U_x$ such that
    \begin{align}
        P_{y^\ast}^{1/2}M_x = \sqrt{r_x}\,U_xD^{1/2}.
        \label{eq:appendix_full_rank_polar}
    \end{align}
    Therefore,
    \begin{align}
        M_x = \sqrt{r_x}\,P_{y^\ast}^{-1/2}U_xD^{1/2}.
        \label{eq:appendix_full_rank_before_normalization}
    \end{align}

    Substituting Eq.~\eqref{eq:appendix_full_rank_before_normalization} into the completeness relation gives
    \begin{align}
        I &= \sum_{x\in X}M_x^\dagger M_x \notag \\
        &= D^{1/2}\left(\sum_{x\in X}r_xU_x^\dagger P_{y^\ast}^{-1}U_x\right)D^{1/2}.
        \label{eq:appendix_full_rank_completeness}
    \end{align}
    Since $D>0$, Eq.~\eqref{eq:appendix_full_rank_completeness} implies
    \begin{align}
        D^{-1} = \sum_{x\in X}r_xU_x^\dagger P_{y^\ast}^{-1}U_x.
    \end{align}
    Taking the unique positive square root gives
    \begin{align}
        D^{1/2} = \left(\sum_{x\in X}r_xU_x^\dagger P_{y^\ast}^{-1}U_x\right)^{-1/2}.
    \end{align}
    Substitution into Eq.~\eqref{eq:appendix_full_rank_before_normalization} yields
    \begin{align}
        M_x = \sqrt{r_x}\,P_{y^\ast}^{-1/2}U_x\left(\sum_{z\in X}r_zU_z^\dagger P_{y^\ast}^{-1}U_z\right)^{-1/2},
    \end{align}
    which is Eq.~\eqref{eq:full_rank_postselection_family}.

    We next prove the sufficiency.
    Let $r_x>0$ and let $U_x$ be unitary for every $x\in X$, and define $M_x$ by
    Eq.~\eqref{eq:full_rank_postselection_family}.
    Every factor appearing in $M_x$ is regular, so the instrument is regular.
    Moreover,
    \begin{align}
        \sum_{x\in X}M_x^\dagger M_x &= \left(\sum_{z\in X}r_zU_z^\dagger P_{y^\ast}^{-1}U_z\right)^{-1/2} \notag \\
        &\quad\times\left(\sum_{x\in X}r_xU_x^\dagger P_{y^\ast}^{-1}U_x\right)\left(\sum_{z\in X}r_zU_z^\dagger P_{y^\ast}^{-1}U_z\right)^{-1/2} \notag \\
        &= I.
    \end{align}
    Thus $\{M_x\}_{x\in X}$ is a quantum instrument.
    Furthermore,
    \begin{align}
        M_x^\dagger P_{y^\ast}M_x = r_x\left(\sum_{z\in X}r_zU_z^\dagger P_{y^\ast}^{-1}U_z\right)^{-1}.
        \label{eq:appendix_full_rank_branch}
    \end{align}
    Consequently,
    \begin{align}
        M_x^\dagger P_{y^\ast}M_x = \frac{r_x}{\sum_{z\in X}r_z}\sum_{z\in X}M_z^\dagger P_{y^\ast}M_z.
    \end{align}
    Therefore $P_{y^\ast}$ realizes information erasure with $q_x=r_x/\sum_zr_z$.

\end{proof}

\subsection{Proof of Theorem~\ref{thm:informationally_complete_erasure}}

\label{app:informationally_complete_erasure}

\begin{proof}

    Let
    \begin{align}
        A := P_{y^\ast}^{-1}, \qquad S := \sum_{z\in X}r_zU_z^\dagger A U_z .
    \end{align}
    Since $P_{y^\ast}>0$ and $r_z>0$, we have $S>0$.

    For the family in Eq.~\eqref{eq:full_rank_postselection_family}, the associated POVM effects are
    \begin{align}
        E_x = r_xS^{-1/2}U_x^\dagger A U_xS^{-1/2}.
        \label{eq:ic_erasure_effects}
    \end{align}
    Consider the real-linear map
    \begin{align}
        & \Phi_S: \operatorname{Herm}(\mathcal{H}) \longrightarrow \operatorname{Herm}(\mathcal{H}), \notag \\
        & \Phi_S(B) := S^{-1/2}BS^{-1/2}.
    \end{align}
    Since $r_x>0$ for every $x\in X$, multiplication by $r_x$ does not affect the real linear span.
    Thus,
    \begin{align}
        E_x = r_x \Phi_S (U_x^\dagger A U_x).
    \end{align}
    Therefore,
    \begin{align}
        \operatorname{span}_{\mathbb{R}}\{E_x\}_{x\in X} &= \operatorname{span}_{\mathbb{R}}\{\Phi_S (U_x^\dagger A U_x)\}_{x\in X} \notag \\
        &= \Phi_S\operatorname{span}_{\mathbb{R}}\{U_x^\dagger A U_x \}_{x\in X}.
    \end{align}
    Since $\Phi_S$ is invertible and
    $\Phi_S^{-1}(\operatorname{Herm}(\mathcal{H}))=\operatorname{Herm}(\mathcal{H})$,
    \begin{align}
        \operatorname{span}_{\mathbb{R}}\{E_x\}_{x\in X} = \operatorname{Herm}(\mathcal{H}) \notag \\
        \Longleftrightarrow\quad \operatorname{span}_{\mathbb{R}} \{U_x^\dagger A U_x\}_{x\in X} = \operatorname{Herm}(\mathcal{H}).
        \label{eq:ic_erasure_span_equivalence}
    \end{align}
    Thus, it suffices to show that there exists a choice of unitaries $\{U_x\}_{x\in X}$ such that
    \begin{align}
        \operatorname{span}_{\mathbb{R}} \{U_x^\dagger A U_x\}_{x\in X} = \operatorname{Herm}(\mathcal{H}).
    \end{align}

    We first prove the necessity.
    If $\{E_x\}_{x\in X}$ is informationally complete, then its effects must span the $d^2$-dimensional real
    vector space $\operatorname{Herm}(\mathcal{H})$.
    Hence
    \begin{align}
        |X|\geq d^2.
    \end{align}

    Suppose next that $P_{y^\ast}\propto I$.
    Then $A\propto I$, and hence
    \begin{align}
        U_x^\dagger A U_x = A \qquad \forall x\in X.
    \end{align}
    Thus the right-hand side of Eq.~\eqref{eq:ic_erasure_span_equivalence} is one-dimensional, and cannot
    equal $\operatorname{Herm}(\mathcal{H})$ for $d\geq2$.
    Therefore,
    \begin{align}
        P_{y^\ast}\not\propto I
    \end{align}
    is necessary.

    We now prove the sufficiency.
    Assume
    \begin{align}
        P_{y^\ast}\not\propto I, \qquad |X|\geq d^2.
    \end{align}
    Here, define
    \begin{align}
        \mathcal{W} := \operatorname{span}_{\mathbb{R}} \{U^\dagger A U\mid U\in U(d)\}\subseteq \operatorname{Herm}(\mathcal{H}).
    \end{align}
    First, we show that
    \begin{align}
        \mathcal{W} = \operatorname{Herm}(\mathcal{H}).
        \label{eq:unitary_orbit_full_span}
    \end{align}
    We then extract $d^2$ unitary conjugates from this orbit.

    Let $\mathcal{W}^\perp$ denote the orthogonal complement of $\mathcal{W}$ with respect to the
    Hilbert--Schmidt inner product, so that
    \begin{align}
        \operatorname{Herm}(\mathcal{H}) = \mathcal{W}\oplus\mathcal{W}^\perp.
    \end{align}
    We show that
    \begin{align}
        \mathcal{W}^\perp = \{0\}.
    \end{align}

    Take an arbitrary $B\in\mathcal{W}^\perp$.
    By definition,
    \begin{align}
        \operatorname{Tr}\!\left(BU^\dagger A U\right) = 0 \qquad \text{for every unitary }U.
    \end{align}
    Since $B$ is Hermitian, there exists a unitary $V$ such that
    \begin{align}
        V^\dagger B V = \sum_{i = 1}^d b_i\ketbra{i}.
    \end{align}
    Moreover, for every unitary $W$,
    \begin{align}
        \operatorname{Tr}\!\left[V^\dagger B V\,W^\dagger A W\right] &= \operatorname{Tr}\!\left[B\,V W^\dagger A W V^\dagger\right] \\
        &= \operatorname{Tr}\!\left[B\,U^\dagger A U\right] = 0,
    \end{align}
    where $U:=WV^\dagger$ is unitary.
    Hence, without loss of generality, we may work in an eigenbasis of $B$ and write
    \begin{align}
        B = \sum_{i = 1}^d b_i\ketbra{i}.
    \end{align}

    Fix arbitrary $i\neq j$.
    Since $A$ is Hermitian and non-scalar, $A$ has at least two distinct eigenvalues, which we denote by
    $\lambda$ and $\mu$.
    We may choose two unitary conjugates $A_{ij}$ and $A_{ji}$ of $A$ that are diagonal in the above basis and
    coincide on all diagonal entries except the $i$th and $j$th entries.
    We choose them such that $A_{ij}$ has $\lambda$ and $\mu$ at the $i$th and $j$th entries, respectively,
    whereas $A_{ji}$ has these two eigenvalues exchanged.
    Hence
    \begin{align}
        A_{ij}-A_{ji} = (\lambda-\mu) \left(\ketbra{i}-\ketbra{j}\right).
    \end{align}
    Since $A_{ij}-A_{ji} \in \mathcal{W}$ and by definition,
    \begin{align}
        0 = \operatorname{Tr}\!\left[B(A_{ij}-A_{ji})\right] = (\lambda-\mu)(b_i-b_j).
    \end{align}
    Since $\lambda\neq\mu$, we have
    \begin{align}
        b_i = b_j.
    \end{align}
    Because $i$ and $j$ are arbitrary, all eigenvalues of $B$ are equal, and therefore
    \begin{align}
        B = cI
    \end{align}
    for some $c\in\mathbb{R}$.
    Since $A\in\mathcal{W}$ and $B\in\mathcal{W}^\perp$,
    \begin{align}
        0 = \operatorname{Tr}(BA) = c\,\operatorname{Tr}(A).
    \end{align}
    Because $A>0$, we have $\operatorname{Tr}(A)>0$, and hence
    \begin{align}
        c = 0.
    \end{align}
    Thus every element of $\mathcal{W}^\perp$ is zero, so
    \begin{align}
        \mathcal{W}^\perp = \{0\}.
    \end{align}
    Since $\mathcal{W}$ is a subspace of the finite-dimensional inner-product space
    $\operatorname{Herm}(\mathcal{H})$, this implies Eq.~\eqref{eq:unitary_orbit_full_span}.

    Since
    \begin{align}
        \dim_{\mathbb{R}}\operatorname{Herm}(\mathcal{H}) = d^2,
    \end{align}
    there exist $d^2$ unitaries $U_1,\ldots,U_{d^2}$ such that
    \begin{align}
        \{U_x^\dagger A U_x\}_{x = 1}^{d^2}
    \end{align}
    is linearly independent over $\mathbb{R}$.
    If $|X|>d^2$, the remaining unitaries may be chosen arbitrarily.
    Therefore,
    \begin{align}
        \operatorname{span}_{\mathbb{R}} \{U_x^\dagger A U_x\}_{x\in X} = \operatorname{Herm}(\mathcal{H}).
    \end{align}
    Equation~\eqref{eq:ic_erasure_span_equivalence} then implies
    \begin{align}
        \operatorname{span}_{\mathbb{R}}\{E_x\}_{x\in X} = \operatorname{Herm}(\mathcal{H}),
    \end{align}
    so the associated POVM is informationally complete.

\end{proof}

\subsection{Proof of Proposition~\ref{prop:projective_quantum_imprint}}
\label{app:projective_imprint}

We next prove Proposition~\ref{prop:projective_quantum_imprint}.
\begin{proof}

    Suppose that there is a postselection operator $P_{y^\ast}$ that realizes information erasure.

    Since $\{M_x\}_{x\in X}$ is a single-Kraus instrument, $\sum_{z\in X}M_z^\dagger M_z=I$.
    For any fixed $x\in X$, we therefore have
    \begin{align}
        \sum_{z\in X\setminus\{x\}}M_z^\dagger M_z = I-M_x^\dagger M_x.
    \end{align}
    It follows that
    \begin{align}
        & \sum_{z\in X\setminus\{x\}}\left(M_z\sqrt{M_x^\dagger M_x}\right)^\dagger\left(M_z\sqrt{M_x^\dagger M_x}\right) \notag \\
        & \qquad = \sqrt{M_x^\dagger M_x}\left(I-M_x^\dagger M_x\right)\sqrt{M_x^\dagger M_x} \notag \\
        & \qquad = M_x^\dagger M_x-\left(M_x^\dagger M_x\right)^2 = 0.
        \label{eq:appendix_projective_positive_sum}
    \end{align}
    Every term in the first line of Eq.~\eqref{eq:appendix_projective_positive_sum} is positive semidefinite.
    Hence
    \begin{align}
        M_z\sqrt{M_x^\dagger M_x} = 0\qquad\forall z\in X\setminus\{x\}.
    \end{align}
    Since $M_x^\dagger M_x$ is a projection, $\sqrt{M_x^\dagger M_x}=M_x^\dagger M_x$.
    Taking the adjoint as well, we obtain
    \begin{align}
        M_zM_x^\dagger M_x = M_x^\dagger M_xM_z^\dagger = 0\qquad\forall z\in X\setminus\{x\}.
        \label{eq:appendix_projective_cross_branch}
    \end{align}

    We next derive two identities used below.
    Let $d_x=\operatorname{rank}M_x$ and consider a singular-value decomposition
    \begin{align}
        M_x = \sum_{j = 1}^{d_x}\sigma_j\ketbra{a_j}{b_j},\qquad\sigma_j>0,
        \label{eq:appendix_projective_svd}
    \end{align}
    where $\{\ket{a_j}\}_{j=1}^{d_x}$ and $\{\ket{b_j}\}_{j=1}^{d_x}$ are orthonormal families.
    Then
    \begin{align}
        M_x^\dagger M_x = \sum_{j = 1}^{d_x}\sigma_j^2\ketbra{b_j}.
    \end{align}
    Because $M_x^\dagger M_x$ is a projection, each nonzero eigenvalue $\sigma_j^2$ equals $1$.
    Hence $\sigma_j=1$ for every $j=1,\ldots,d_x$.
    Therefore,
    \begin{align}
        M_xM_x^\dagger M_x = M_x,\qquad M_x^\dagger M_xM_x^\dagger = M_x^\dagger.
        \label{eq:appendix_projective_partial_isometry}
    \end{align}
    By Lemma~\ref{lem:reference_branch_condition}, there exist $x^\ast\in X$ and positive numbers
    $\{r_z\}_{z\in X}$ satisfying $r_{x^\ast}=1$ and
    \begin{align}
        M_z^\dagger P_{y^\ast}M_z = r_zM_{x^\ast}^\dagger P_{y^\ast}M_{x^\ast}\qquad\forall z\in X.
        \label{eq:appendix_projective_reference}
    \end{align}

    Fix $x\in X$.
    We consider separately the cases $x=x^\ast$ and $x\neq x^\ast$.

    First, suppose that $x=x^\ast$.
    Since $|X|\geq2$, choose $z\in X\setminus\{x\}$.
    Equation~\eqref{eq:appendix_projective_reference} gives
    \begin{align}
        M_z^\dagger P_{y^\ast}M_z = r_zM_x^\dagger P_{y^\ast}M_x.
    \end{align}
    Multiplying this equality from the left by $M_x^\dagger M_x$, and using
    Eqs.~\eqref{eq:appendix_projective_cross_branch} and \eqref{eq:appendix_projective_partial_isometry},
    gives
    \begin{align}
        0 = r_zM_x^\dagger P_{y^\ast}M_x.
    \end{align}
    Since $r_z>0$,
    \begin{align}
        M_x^\dagger P_{y^\ast}M_x = 0.
        \label{eq:appendix_projective_reference_zero}
    \end{align}
    Equation~\eqref{eq:appendix_projective_reference} then implies
    \begin{align}
        M_z^\dagger P_{y^\ast}M_z = 0\qquad\forall z\in X.
        \label{eq:appendix_projective_all_zero_case_one}
    \end{align}

    Second, suppose that $x\neq x^\ast$.
    Taking $z=x$ in Eq.~\eqref{eq:appendix_projective_reference}, we have
    \begin{align}
        M_x^\dagger P_{y^\ast}M_x = r_xM_{x^\ast}^\dagger P_{y^\ast}M_{x^\ast}.
    \end{align}
    Multiplication from the left by $M_x^\dagger M_x$, together with
    Eqs.~\eqref{eq:appendix_projective_cross_branch} and \eqref{eq:appendix_projective_partial_isometry},
    gives
    \begin{align}
        M_x^\dagger P_{y^\ast}M_x = 0.
    \end{align}
    Since $r_x>0$, it follows that $M_{x^\ast}^\dagger P_{y^\ast}M_{x^\ast}=0$.
    Equation~\eqref{eq:appendix_projective_reference} therefore again gives
    \begin{align}
        M_z^\dagger P_{y^\ast}M_z = 0\qquad\forall z\in X.
        \label{eq:appendix_projective_all_zero_case_two}
    \end{align}

    Thus, in either case,
    \begin{align}
        \sum_{z\in X}M_z^\dagger P_{y^\ast}M_z = 0.
    \end{align}
    This contradicts $P_{y^\ast}\in\operatorname{Adm}(\mathcal{I})$.

    Therefore, projective measurements are quantum imprints.

\end{proof}

\subsection{Proof of Proposition~\ref{prop:regular_imprint_arbitrary_dimension}}
\label{app:regular_imprint_arbitrary_dimension}
We now prove Proposition~\ref{prop:regular_imprint_arbitrary_dimension}.

\begin{proof}

    First, we verify that $\{M_x\}_{x=0}^2$ is regular.
    Since $A_0=I$, $A_1=I-aS_d$, and $A_2=I+bS_d^\dagger$ are triangular matrices with all diagonal entries
    equal to $1$, each $A_x$ is invertible.
    Hence $\left(\sum_{z=0}^{2}A_z^\dagger A_z\right)^{-1/2}$ is also invertible.
    Therefore, each $M_x$ is invertible, and the instrument $\{M_x\}_{x=0}^2$ is regular.

    Suppose that information erasure is realized.
    Then we get the following from Lemma~\ref{lem:reference_branch_condition}:
    \begin{align}
        A_i^\dagger P_{y^\ast}A_i = r_i P_{y^\ast},\quad i = 1,2.
    \end{align}
    Let us first suppose that the nonzero postselection effect $P_{y^\ast}$ is rank-deficient and take
    $0\neq \ket{v}\in \ker{P}_{y^\ast}$.
    From the above equation, we have
    \begin{align}
        \bra{v}A_i^\dagger P_{y^\ast}A_i\ket{v} = r_i \bra{v}P_{y^\ast}\ket{v} = 0.
    \end{align}
    Since $P_{y^\ast}\geq 0$,
    \begin{align}
        \bra{v}A_i^\dagger P_{y^\ast}A_i\ket{v} = \| \sqrt{P}_{y^\ast}A_i\ket{v} \|^2 = 0, \\
        \implies \sqrt{P}_{y^\ast}A_i\ket{v} = 0 \implies P_{y^\ast}A_i\ket{v} = 0.
    \end{align}
    By substituting this into the original form of $A_i$, we have
    \begin{align}
        P_{y^\ast}(I-a S_d)\ket{v} = 0 \implies P_{y^\ast}S_d\ket{v} = 0,\label{Eq.kernel} \\
        P_{y^\ast}(I+b S_d^\dagger)\ket{v} = 0 \implies P_{y^\ast}S_d^\dagger\ket{v} = 0.
    \end{align}
    Here, let us decompose $\ket{v}$ in the basis $\{\ket{j}\}_{j=1}^d$ as
    \begin{align}
        \ket{v} = \sum_{j = 1}^d c_j\ket{j},
    \end{align}
    and let $m$ be the largest index such that $c_m\neq0$:
    \begin{align}
        \ket{v} = \sum_{j = 1}^m c_j\ket{j}.
    \end{align}
    Then since $S_d\ket{j}=\ket{j-1}$, we have
    \begin{align}
        S_d^{m-1}\ket{v} = c_m\ket{1}.
        \label{Eq.min_condition}
    \end{align}
    Since Eq.~\eqref{Eq.kernel} implies that $S_d\ket{w}\in\ker P_{y^\ast}$ whenever
    $\ket{w}\in\ker P_{y^\ast}$, iterating Eq.~\eqref{Eq.kernel} $m-1$ times gives
    \begin{align}
        P_{y^\ast}S_d^{m-1}\ket{v} = 0.
    \end{align}
    Therefore, by taking $P_{y^\ast}$ on both sides of Eq.~\eqref{Eq.min_condition}, we have
    \begin{align}
        P_{y^\ast}\ket{1} = 0.
    \end{align}

    Hence $\ket{1}\in\ker P_{y^\ast}$.
    By taking $\ket{v}$ as $\ket{1}$ we have
    \begin{align}
        S_d^{\dagger j-1}\ket{1} = \ket{j},\quad P_{y^\ast}S_d^{\dagger j-1}\ket{1} = 0.
    \end{align}
    Therefore, we have
    \begin{align}
        \ket{j}\in\ker P_{y^\ast},\qquad j = 1,\ldots,d.
    \end{align}
    Thus,
    \begin{align}
        \ker P_{y^\ast} = \mathcal{H},
    \end{align}
    which implies
    \begin{align}
        P_{y^\ast} = 0.
    \end{align}
    This contradicts the assumption that $P_{y^\ast}\neq0$.
    Thus no nonzero rank-deficient postselection effect can realize information erasure.

    Let us now consider the case where $P_{y^\ast}$ is full-rank.
    As we show in Proposition~\ref{prop:full_rank_postselection_characterization}, the necessary and
    sufficient condition for information erasure of $\{M_x\}_{x\in X}$ is below:
    \begin{align}
        M_x = \sqrt{r_x}\,P_{y^\ast}^{-1/2}U_x\left(\sum_{z\in X}r_z U_z^\dagger P_{y^\ast}^{-1}U_z\right)^{-1/2},
    \end{align}
    for arbitrary unitary operators $\{U_x\}_{x\in X}$ and positive numbers $\{r_x\}_{x\in X}$.
    Here, from the definition of $A_x$, we have
    \begin{align}
        M_xM_0^{-1} = A_x.
    \end{align}
    On the other hand, the characterization in Proposition~\ref{prop:full_rank_postselection_characterization}
    gives
    \begin{align}
        M_xM_0^{-1} = \sqrt{\frac{r_x}{r_0}}\,P_{y^\ast}^{-1/2}U_xU_0^\dagger P_{y^\ast}^{1/2}.
    \end{align}
    Therefore,
    \begin{align}
        P_{y^\ast}^{1/2}A_xP_{y^\ast}^{-1/2} = \sqrt{\frac{r_x}{r_0}}\,U_xU_0^\dagger.
    \end{align}
    Since $U_xU_0^\dagger$ is unitary and $P_{y^\ast}>0$, $A_x$ must be at least diagonalizable.

    Since $A_1$ is upper triangular with all diagonal entries equal to $1$, its only eigenvalue is $1$.
    If $A_1$ were diagonalizable, it would therefore be similar to the identity matrix, which would imply
    $A_1=I$.
    However, $A_1=I-aS_d\neq I$ since $a>0$ and $S_d\neq0$.
    Hence $A_1$ is not diagonalizable.

    Hence, the quantum instrument given in this proposition is a quantum imprint.

\end{proof}

\subsection{Proof of Theorem~\ref{thm:same_povm_erasure_imprint}}
\label{app:same_povm_erasure_imprint}

We finally prove Theorem~\ref{thm:same_povm_erasure_imprint} using
Propositions~\ref{prop:regular_erasure} and \ref{prop:regular_imprint_arbitrary_dimension}.

\begin{proof}

    Let $\{M_x^{\mathrm{imp}}\}_{x=0}^2$ be the regular quantum imprint provided by
    Proposition~\ref{prop:regular_imprint_arbitrary_dimension}, and define a POVM by
    \begin{align}
        E_x = (M_x^{\mathrm{imp}})^\dagger M_x^{\mathrm{imp}}\qquad x\in\{0,1,2\}.
    \end{align}
    Every $M_x^{\mathrm{imp}}$ is regular, so $E_x>0$.

    Fix a unit vector $\ket{v}$.
    For each $x\in\{0,1,2\}$, define
    \begin{align}
        \ket{u_x} = \frac{E_x^{-1/2}\ket{v}}{\|E_x^{-1/2}\ket{v}\|}.
        \label{eq:appendix_same_povm_ux}
    \end{align}
    Both $\ket{u_x}$ and $\ket{v}$ are unit vectors, so there exists a unitary $U_x$ satisfying
    $U_x^\dagger\ket{v}=\ket{u_x}$.
    Define
    \begin{align}
        M_x^{\mathrm{er}} = U_xE_x^{1/2}.
        \label{eq:appendix_same_povm_erasure_realization}
    \end{align}
    Then since $E_x >0$, every $M_x^{\mathrm{er}}$ is invertible and
    \begin{align}
        (M_x^{\mathrm{er}})^\dagger M_x^{\mathrm{er}} = E_x.
    \end{align}
    Therefore, the two instruments $\{M_x^{\mathrm{imp}}\}_{x=0}^2,\{M_x^{\mathrm{er}}\}_{x=0}^2$ realize
    exactly the same POVM.

    Next, we show that $\{M_x^{\mathrm{er}}\}_{x=0}^2$ is an information-erasure-inducing instrument.
    Using Eq.~\eqref{eq:appendix_same_povm_ux}, we obtain
    \begin{align}
        (M_x^{\mathrm{er}})^\dagger\ket{v} &= E_x^{1/2}U_x^\dagger\ket{v} = E_x^{1/2}\ket{u_x} \notag \\
        &= \frac{\ket{v}}{\|E_x^{-1/2}\ket{v}\|}.
        \label{eq:appendix_same_povm_action}
    \end{align}
    Taking $x^\ast=0$ and using Eq.~\eqref{eq:appendix_same_povm_action} gives
    \begin{align}
        \left(M_x^{\mathrm{er}}(M_0^{\mathrm{er}})^{-1}\right)^\dagger\ket{v} = \frac{\|E_0^{-1/2}\ket{v}\|}{\|E_x^{-1/2}\ket{v}\|}\ket{v}\qquad x\in\{0,1,2\}.
        \label{eq:appendix_same_povm_common_eigenvector}
    \end{align}
    Thus the operators $\{(M_x^{\mathrm{er}}(M_0^{\mathrm{er}})^{-1})^\dagger\}_{x=0}^2$ have the common
    eigenvector $\ket{v}$.
    By Proposition~\ref{prop:regular_erasure}, $\{M_x^{\mathrm{er}}\}_{x=0}^2$ is an
    information-erasure-inducing instrument.

    On the other hand, $\{M_x^{\mathrm{imp}}\}_{x=0}^2$ is a quantum imprint by
    Proposition~\ref{prop:regular_imprint_arbitrary_dimension}.
    Therefore the POVM $\{E_x\}_{x=0}^2$ admits two regular single-Kraus realizations with opposite
    information structures.

\end{proof}

\section{Proof of Theorem~\ref{thm:first_order_spam_separation}}

\label{app:first_order_spam_separation}
\begin{proof}
    For the measurement pair $(\boldsymbol{M},P_{y^\ast})$, define
    \begin{align}
        F_x := M_x^\dagger P_{y^\ast}M_x,\qquad F := \sum_xF_x.
    \end{align}
    The ideal posterior distribution is then given by
    \begin{align}
        p(x| y^\ast,\rho) = \frac{\operatorname{Tr}(\rho F_x)}{\operatorname{Tr}(\rho F)}.
    \end{align}
    Under the noise model in Eq.~\eqref{eq:spam_noise_model},
    \begin{align}
        \widetilde{M}_x &= M_x+\delta M_x+O(\epsilon^2), \\
        \widetilde{P}_{y^\ast} &= P_{y^\ast}+\delta P_{y^\ast}+O(\epsilon^2).
    \end{align}
    Therefore, the operator including the postselection effect becomes
    \begin{align}
        \widetilde{F}_x &:= \widetilde{M}_x^\dagger\widetilde{P}_{y^\ast}\widetilde{M}_x \notag \\
        &= \left(M_x^\dagger+\delta M_x^\dagger\right) \left(P_{y^\ast}+\delta P_{y^\ast}\right) \left(M_x+\delta M_x\right)+O(\epsilon^2) \notag \\
        &= M_x^\dagger P_{y^\ast}M_x +\delta M_x^\dagger P_{y^\ast}M_x +M_x^\dagger P_{y^\ast}\delta M_x \notag \\
        &\quad +M_x^\dagger\delta P_{y^\ast}M_x +O(\epsilon^2) \notag \\
        &= F_x+\delta F_x+\delta F_x^P+O(\epsilon^2),
    \end{align}
    where
    \begin{align}
        \delta F_x &:= \delta M_x^\dagger P_{y^\ast}M_x +M_x^\dagger P_{y^\ast}\delta M_x, \\
        \delta F_x^P &:= M_x^\dagger\delta P_{y^\ast}M_x.
    \end{align}
    Similarly,
    \begin{align}
        \widetilde{F} := \sum_x\widetilde{F}_x = F+\delta F+\delta F^P+O(\epsilon^2),
    \end{align}
    where
    \begin{align}
        \delta F := \sum_x\delta F_x,\qquad \delta F^P := \sum_x\delta F_x^P.
    \end{align}
    The noisy input state is written as
    \begin{align}
        \widetilde{\rho} = \rho+\delta\rho+O(\epsilon^2).
    \end{align}
    Hence, the implemented posterior distribution is
    \begin{align}
        \widetilde{p}(x| y^\ast,\widetilde{\rho}) = \frac{\operatorname{Tr}(\widetilde{\rho}\,\widetilde{F}_x)}{\operatorname{Tr}(\widetilde{\rho}\,\widetilde{F})}.
    \end{align}
    Expanding the numerator to first order gives
    \begin{align}
        \operatorname{Tr}(\widetilde{\rho}\,\widetilde{F}_x) &= \operatorname{Tr}(\rho F_x) +\operatorname{Tr}(\delta\rho F_x) +\operatorname{Tr}(\rho\delta F_x) \notag \\
        &\quad +\operatorname{Tr}(\rho\delta F_x^P) +O(\epsilon^2).
    \end{align}
    Here, $O(\epsilon^2)$ collectively denotes all terms of second and higher order in the perturbations.
    Similarly, the denominator becomes
    \begin{align}
        \operatorname{Tr}(\widetilde{\rho}\,\widetilde{F}) &= \operatorname{Tr}(\rho F) +\operatorname{Tr}(\delta\rho F) +\operatorname{Tr}(\rho\delta F) \notag \\
        &\quad +\operatorname{Tr}(\rho\delta F^P) +O(\epsilon^2).
    \end{align}
    Let
    \begin{align}
        N_x := \operatorname{Tr}(\rho F_x),\qquad D := \operatorname{Tr}(\rho F),
    \end{align}
    and define their first-order variations by
    \begin{align}
        \delta N_x &:= \operatorname{Tr}(\delta\rho F_x) +\operatorname{Tr}(\rho\delta F_x) +\operatorname{Tr}(\rho\delta F_x^P), \\
        \delta D &:= \operatorname{Tr}(\delta\rho F) +\operatorname{Tr}(\rho\delta F) +\operatorname{Tr}(\rho\delta F^P).
    \end{align}
    Then
    \begin{align}
        \widetilde{p}(x| y^\ast,\widetilde{\rho}) = \frac{N_x+\delta N_x}{D+\delta D}+O(\epsilon^2).
    \end{align}
    Using the first-order expansion of the ratio,
    \begin{align}
        \frac{N_x+\delta N_x}{D+\delta D} = \frac{N_x}{D} +\frac{1}{D} \left(\delta N_x-\frac{N_x}{D}\delta D\right) +O(\epsilon^2).
    \end{align}
    Therefore,
    \begin{align}
        \widetilde{p}(x| y^\ast,\widetilde{\rho}) &= p(x| y^\ast,\rho) +\frac{\delta N_x -p(x| y^\ast,\rho)\delta D}{\operatorname{Tr}(\rho F)} +O(\epsilon^2) \notag \\
        &= p(x| y^\ast,\rho) +\frac{\operatorname{Tr}\!\left[\delta\rho \left(F_x-p(x| y^\ast,\rho)F\right)\right]}{\operatorname{Tr}(\rho F)} \notag \\
        &\quad +\frac{\operatorname{Tr}\!\left[\rho \left(\delta F_x-p(x| y^\ast,\rho)\delta F\right)\right]}{\operatorname{Tr}(\rho F)} \notag \\
        &\quad +\frac{\operatorname{Tr}\!\left[\rho \left(\delta F_x^P-p(x| y^\ast,\rho)\delta F^P\right)\right]}{\operatorname{Tr}(\rho F)} +O(\epsilon^2).
    \end{align}
    Since $(\boldsymbol{M},P_{y^\ast})$ realizes information erasure, the ideal posterior distribution is
    independent of the input state and satisfies
    \begin{align}
        p(x| y^\ast,\rho) = q_x.
    \end{align}
    Equivalently, the information-erasure condition gives
    \begin{align}
        F_x = q_xF.
    \end{align}
    Therefore, the first-order contribution arising from the preparation perturbation vanishes:
    \begin{align}
        \operatorname{Tr}\!\left[\delta\rho \left(F_x-q_xF\right)\right] = 0.
    \end{align}
    Hence,
    \begin{align}
        \widetilde{p}(x| y^\ast,\widetilde{\rho}) &= q_x +\frac{\operatorname{Tr}\!\left[\rho \left(\delta F_x-q_x\delta F\right)\right]}{\operatorname{Tr}(\rho F)} \notag \\
        &\quad +\frac{\operatorname{Tr}\!\left[\rho \left(\delta F_x^P-q_x\delta F^P\right)\right]}{\operatorname{Tr}(\rho F)} +O(\epsilon^2).
    \end{align}
    Using
    \begin{align}
        R_x := \delta F_x-q_x\delta F,\qquad R_x^P := \delta F_x^P-q_x\delta F^P,
    \end{align}
    we obtain
    \begin{align}
        \widetilde{p}(x| y^\ast,\widetilde{\rho}) = q_x +\frac{\operatorname{Tr}\!\left[\rho(R_x+R_x^P)\right]}{\operatorname{Tr}(\rho F)} +O(\epsilon^2).
    \end{align}
    Finally, by the assumption of Theorem~\ref{thm:first_order_spam_separation},
    \begin{align}
        R_x^P = 0,\qquad \forall x\in X.
    \end{align}
    Therefore,
    \begin{align}
        \widetilde{p}(x| y^\ast,\widetilde{\rho}) = q_x +\frac{\operatorname{Tr}(\rho R_x)}{\operatorname{Tr}(\rho F)} +O(\epsilon^2).
    \end{align}

\end{proof}

\section{Qubit illustration of SPAM separation}
\label{app:minimal_spam_sim}

In this appendix, we give a single qubit example illustrating first-order SPAM separation.
All error parameters are of order $O(\epsilon)$ unless stated otherwise.

Let $X$, $Y$, and $Z$ denote the Pauli operators.
Consider the single-Kraus instrument defined by
\begin{align}
    M_0 = \left(\frac12 I-\frac14 X\right)^{1/2}, \qquad M_1 = M_0Z.
    \label{eq:kernel_qubit_instrument}
\end{align}
We choose the postselection effect $P_{y^\ast}=\frac12 I+\frac14 X$.
Then the pair $(\{M_0,M_1\},P_{y^\ast})$ realizes information erasure with $q_0=q_1=1/2$.
Hence, the posterior distribution is independent of the input state.

\subsection{Postselection-noise kernel}

For an arbitrary Hermitian perturbation $\delta P_{y^\ast}$, define
$H:=P_{y^\ast}^{-1/2}\delta P_{y^\ast}P_{y^\ast}^{-1/2}$.
Substituting the instrument into the postselection residual gives
\begin{align}
    R_0^P = \frac{3}{32}(H-ZHZ), \qquad R_1^P = -R_0^P.
    \label{eq:kernel_qubit_postselection_residual}
\end{align}
Hence, $R_x^P=0$ for both outcomes if and only if $[H,Z]=0$.
For a qubit, this is equivalent to $H=aI+bZ$ with $a,b\in\mathbb{R}$.
Therefore,
\begin{align}
    \mathcal{K}_P = \left\{P_{y^\ast}^{1/2}(aI+bZ)P_{y^\ast}^{1/2} \,\middle|\, a,b\in\mathbb{R}\right\}.
    \label{eq:kernel_qubit_noise_space}
\end{align}
Thus, the postselection-noise kernel is two-dimensional and contains, in addition to the uniform rescaling
direction, an independent non-scaling direction.

To illustrate this non-scaling direction, consider
\begin{align}
    \widetilde{P}_{y^\ast} = P_{y^\ast} + P_{y^\ast}^{1/2}(bZ)P_{y^\ast}^{1/2},
    \label{eq:kernel_qubit_random_effect}
\end{align}
for $b$ such that $\widetilde{P}_{y^\ast}$ remains a valid POVM effect.

The same direction also has a direct calibration interpretation.
Let $U_y(\vartheta_P)=e^{-i\vartheta_PY/2}$ and consider a small unitary rotation of the postselection effect,
\begin{align}
    \widetilde{P}_{y^\ast} &= U_y(\vartheta_P) P_{y^\ast} U_y(\vartheta_P)^\dagger \notag \\
    &= P_{y^\ast} -\frac14\vartheta_P Z + O(\vartheta_P^2).
\end{align}
Since $P_{y^\ast}^{1/2} Z P_{y^\ast}^{1/2} =(\sqrt{3}/4)Z$, we have
$\delta P_{y^\ast} = P_{y^\ast}^{1/2} [-(\vartheta_P/\sqrt{3})Z] P_{y^\ast}^{1/2}\in\mathcal{K}_P$.
Hence, a small unitary rotation about the $y$ axis is invisible to first order.
\subsection{First-order SPAM error separation}
We now introduce a coherent measurement error by mixing the two outcome branches of the ideal instrument.
Specifically, let
\begin{align}
    \widetilde{M}_0 &= \cos\epsilon_M\,M_0+\sin\epsilon_M\,M_1, \\
    \widetilde{M}_1 &= -\sin\epsilon_M\,M_0+\cos\epsilon_M\,M_1 .
\end{align}
This transformation preserves the completeness relation and therefore defines a valid single-Kraus instrument.
The ideal postselection effect $P_{y^\ast}$ is kept fixed.

For the initial state, let $|\psi_\varphi\rangle = \cos(\varphi/2)|0\rangle+\sin(\varphi/2)|1\rangle$.
We take the ideal state to be $\rho=\ketbra{\psi_\vartheta}$ and introduce a state-preparation error as
$\widetilde{\rho}=\ketbra{\psi_{\vartheta+\epsilon_S}}$.

Without postselection, the probabilities are
\begin{align}
    \widetilde{p}(x|\widetilde{\rho}) &= \frac12 + (-1)^{x+1} \left(\frac14\sin\vartheta + \frac14\cos\vartheta\,\epsilon_S - \cos\vartheta\,\epsilon_M\right) \notag \\
    &\quad+O(\epsilon^2).
    \label{eq:kernel_qubit_raw_expansion}
\end{align}
Thus, both the state-preparation error and the measurement error contribute to the probability distribution at
first order.

Applying Theorem~\ref{thm:first_order_spam_separation}, we obtain
\begin{align}
    \widetilde{p}(x|y^\ast,\widetilde{\rho}) = \frac12 + (-1)^x\cos\vartheta\,\epsilon_M +O(\epsilon^2).
    \label{eq:kernel_qubit_posterior_expansion}
\end{align}
Hence, the first-order contribution of the state-preparation error vanishes, while the measurement error
remains visible.

\section{Feedback Control}
\label{app:feedback_control}

In this section, we provide the mathematical proof needed for the main text.
The goal of this section is to establish Theorem~\ref{thm:feedback_guarantee}.
To this end, we first establish the direct-sum structure of the visible error space that underlies the
feedback decomposition.

\begin{lem}[Decomposition of the visible error space]
    \label{lem:visible_error_decomposition}
    For the noncontrollable visible subspace $V_{\mathrm{nc}}$ defined in the main text,
    $V_{\mathrm{vis}}=V_{\mathrm{fb}}\oplus V_{\mathrm{nc}}$.
    Consequently, since $V=\mathcal{K}\oplus^\perp V_{\mathrm{vis}}$,
    $V=\mathcal{K}\oplus^\perp\left(V_{\mathrm{fb}}\oplus V_{\mathrm{nc}}\right)$.
\end{lem}
Note that the above direct sum is not necessarily orthogonal with respect to the inner product on $V$.
\begin{proof}
    We first verify that $V_{\mathrm{nc}}$ is a subspace of $V_{\mathrm{vis}}$.
    By definition,
    \begin{align}
        V_{\mathrm{nc}} = \{v\in V_{\mathrm{vis}}| \langle T(v),T(w)\rangle_W = 0,\ \forall w\in V_{\mathrm{fb}}\}.
        \label{eq:app_direct_1}
    \end{align}
    Take $v_1,v_2\in V_{\mathrm{nc}}$ and $a,b\in\mathbb{R}$.
    For any $w\in V_{\mathrm{fb}}$, the linearity of $T$ gives
    \begin{align}
        & \langle T(av_1+bv_2),T(w)\rangle_W = a\langle T(v_1),T(w)\rangle_W \notag \\
        &\quad+b\langle T(v_2),T(w)\rangle_W = 0.
        \label{eq:app_direct_2}
    \end{align}
    Since $V_{\mathrm{vis}}$ is a subspace, $av_1+bv_2\in V_{\mathrm{vis}}$.
    Therefore, $av_1+bv_2\in V_{\mathrm{nc}}$, and $V_{\mathrm{nc}}$ is a subspace of $V_{\mathrm{vis}}$.
    In particular, it contains the zero vector.

    We next prove
    \begin{align}
        V_{\mathrm{vis}} = V_{\mathrm{fb}}+V_{\mathrm{nc}}.
        \label{eq:app_direct_sum_first}
    \end{align}
    Since $V_{\mathrm{fb}}\subseteq V_{\mathrm{vis}}$ by assumption and
    $V_{\mathrm{nc}}\subseteq V_{\mathrm{vis}}$ by definition,
    \begin{align}
        V_{\mathrm{fb}}+V_{\mathrm{nc}}\subseteq V_{\mathrm{vis}}.
        \label{eq:app_direct_subset_1}
    \end{align}
    It therefore suffices to show the reverse inclusion
    \begin{align}
        V_{\mathrm{vis}}\subseteq V_{\mathrm{fb}}+V_{\mathrm{nc}}.
        \label{eq:app_direct_subset_2}
    \end{align}

    Take an arbitrary $w\in V_{\mathrm{vis}}$ and define
    \begin{align}
        Y_{\mathrm{fb}} := T(V_{\mathrm{fb}})\subseteq W.
        \label{eq:app_Yfb}
    \end{align}
    Since $Y_{\mathrm{fb}}$ is a subspace of $W$, the standard orthogonal decomposition of a
    finite-dimensional inner-product space gives
    \begin{align}
        W = Y_{\mathrm{fb}}\oplus Y_{\mathrm{fb}}^\perp.
        \label{eq:app_W_decomposition}
    \end{align}
    Hence, $T(w)\in W$ admits the unique decomposition
    \begin{align}
        T(w) = \alpha+\beta,\qquad \alpha\in Y_{\mathrm{fb}},\qquad \beta\in Y_{\mathrm{fb}}^\perp.
        \label{eq:app_Tw_decomposition}
    \end{align}
    Since $\alpha\in Y_{\mathrm{fb}}=T(V_{\mathrm{fb}})$, there exists $u\in V_{\mathrm{fb}}$ such that
    \begin{align}
        \alpha = T(u).
        \label{eq:app_alpha_Tu}
    \end{align}
    Define
    \begin{align}
        v := w-u.
        \label{eq:app_v_definition}
    \end{align}
    Since $w\in V_{\mathrm{vis}}$ and $u\in V_{\mathrm{fb}}\subseteq V_{\mathrm{vis}}$, the fact that
    $V_{\mathrm{vis}}$ is a subspace gives
    \begin{align}
        v = w-u\in V_{\mathrm{vis}}.
        \label{eq:app_v_in_Vvis}
    \end{align}
    Furthermore,
    \begin{align}
        T(v) = T(w-u) = T(w)-T(u) = (\alpha+\beta)-\alpha = \beta.
        \label{eq:app_Tv_beta}
    \end{align}
    For arbitrary $u'\in V_{\mathrm{fb}}$, $T(u')\in Y_{\mathrm{fb}}$, while $\beta\in Y_{\mathrm{fb}}^\perp$.
    Therefore,
    \begin{align}
        \langle T(v),T(u')\rangle_W = \langle\beta,T(u')\rangle_W = 0.
        \label{eq:app_v_orthogonal}
    \end{align}
    Thus, $v\in V_{\mathrm{nc}}$.
    Consequently,
    \begin{align}
        w = u+v,\qquad u\in V_{\mathrm{fb}},\qquad v\in V_{\mathrm{nc}}.
        \label{eq:app_w_decomposition}
    \end{align}
    Since $w\in V_{\mathrm{vis}}$ was arbitrary,
    \begin{align}
        V_{\mathrm{vis}}\subseteq V_{\mathrm{fb}}+V_{\mathrm{nc}},
    \end{align}
    and hence
    \begin{align}
        V_{\mathrm{vis}} = V_{\mathrm{fb}}+V_{\mathrm{nc}}.
        \label{eq:app_sum_equal}
    \end{align}

    Finally, we prove
    \begin{align}
        V_{\mathrm{fb}}\cap V_{\mathrm{nc}} = \{0\}.
        \label{eq:app_intersection_zero}
    \end{align}
    Let $v\in V_{\mathrm{fb}}\cap V_{\mathrm{nc}}$.
    Since $v\in V_{\mathrm{nc}}$, for every $u\in V_{\mathrm{fb}}$,
    \begin{align}
        \langle T(v),T(u)\rangle_W = 0.
    \end{align}
    Since $v\in V_{\mathrm{fb}}$, taking $u=v$ gives
    \begin{align}
        \langle T(v),T(v)\rangle_W = \|T(v)\|_W^2 = 0.
    \end{align}
    Therefore,
    \begin{align}
        T(v) = 0,
    \end{align}
    and hence
    \begin{align}
        v\in\ker T = \mathcal{K}.
    \end{align}
    On the other hand, $v\in V_{\mathrm{fb}}\subseteq V_{\mathrm{vis}}$, so
    \begin{align}
        v\in\mathcal{K}\cap V_{\mathrm{vis}}.
    \end{align}
    Since $V_{\mathrm{vis}}=\mathcal{K}^\perp$,
    \begin{align}
        \mathcal{K}\cap V_{\mathrm{vis}} = \mathcal{K}\cap\mathcal{K}^\perp = \{0\}.
    \end{align}
    Thus, $v=0$, and therefore
    \begin{align}
        V_{\mathrm{fb}}\cap V_{\mathrm{nc}} = \{0\}.
    \end{align}
    Combining the above results gives
    \begin{align}
        V_{\mathrm{vis}} = V_{\mathrm{fb}}\oplus V_{\mathrm{nc}}.
    \end{align}
\end{proof}

As an immediate consequence of Lemma~\ref{lem:visible_error_decomposition}, the effective first-order
perturbation admits the following unique decomposition.

\begin{cor}
    \label{cor:app_feedback_decomposition}
    For every $\eta$, there exist unique $\delta M_K\in K$, $\delta M_{\mathrm{fb}}(\eta)\in V_{\mathrm{fb}}$,
    and $\delta M_{\mathrm{nc}}\in V_{\mathrm{nc}}$ such that
    $\delta M(\eta)=\delta M_K+\delta M_{\mathrm{fb}}(\eta)+\delta M_{\mathrm{nc}}$.
\end{cor}

It is important to emphasize that $\eta$ parametrizes only the controllable directions accessible through the
feedback.
Accordingly, we assume that the first-order effect of the feedback belongs to $V_{\mathrm{fb}}$.
Under this assumption, varying $\eta$ changes only the $V_{\mathrm{fb}}$ component of $\delta M(\eta)$, while
the $K$ and $V_{\mathrm{nc}}$ components remain unchanged.

\begin{proof}
    This follows immediately from Lemma~\ref{lem:visible_error_decomposition} and the assumption that the
    feedback acts only within $V_{\mathrm{fb}}$.
\end{proof}

With this decomposition in hand, we are now ready to prove the feedback guarantee theorem.
\subsection{Proof of Theorem~\ref{thm:feedback_guarantee}}
\label{app:proof_sufficient_decrease}

\begin{proof}
    We first show that $\alpha$ and $\beta$ exist and satisfy
    \begin{align}
        0<\alpha\leq\beta<\infty.
        \label{eq:app_alpha_beta_bounds}
    \end{align}
    Recall that
    \begin{align}
        \alpha := \min_{v\in V_{\mathrm{fb}}\setminus\{0\}}\frac{\|T(v)\|_W}{\|v\|_V},\qquad \beta := \max_{v\in V_{\mathrm{fb}}\setminus\{0\}}\frac{\|T(v)\|_W}{\|v\|_V}.
    \end{align}

    Let $\mathcal{K}:=\ker T$.
    Since
    \begin{align}
        V_{\mathrm{vis}} = \mathcal{K}^\perp,\qquad V_{\mathrm{fb}}\subseteq V_{\mathrm{vis}},
    \end{align}
    if $v\in V_{\mathrm{fb}}$ and $T(v)=0$, then $v\in\mathcal{K}$ and $v\in\mathcal{K}^\perp$.
    Hence
    \begin{align}
        v\in\mathcal{K}\cap\mathcal{K}^\perp = \{0\}.
    \end{align}
    Therefore, the restriction
    \begin{align}
        T|_{V_{\mathrm{fb}}}:V_{\mathrm{fb}}\longrightarrow W
        \label{eq:app_T_restriction}
    \end{align}
    is injective.

    Define
    \begin{align}
        S_{\mathrm{fb}} := \{v\in V_{\mathrm{fb}}| \|v\|_V = 1\}.
        \label{eq:app_Sfb}
    \end{align}
    Since $V_{\mathrm{fb}}\neq\{0\}$, $S_{\mathrm{fb}}$ is nonempty, closed, and bounded.
    Because $V_{\mathrm{fb}}$ is finite dimensional, $S_{\mathrm{fb}}$ is compact.
    The map
    \begin{align}
        v\longmapsto\|T(v)\|_W
    \end{align}
    is continuous, and therefore, by the extreme value theorem, it attains both a minimum and a maximum on
    $S_{\mathrm{fb}}$.
    Hence
    \begin{align}
        \alpha = \min_{v\in S_{\mathrm{fb}}}\|T(v)\|_W,\qquad \beta = \max_{v\in S_{\mathrm{fb}}}\|T(v)\|_W.
    \end{align}
    Since $T|_{V_{\mathrm{fb}}}$ is injective, $T(v)\neq0$ for every $v\in S_{\mathrm{fb}}$.
    Therefore,
    \begin{align}
        \|T(v)\|_W>0
    \end{align}
    for all $v\in S_{\mathrm{fb}}$, and in particular
    \begin{align}
        \alpha>0.
    \end{align}
    Since the maximum exists,
    \begin{align}
        \beta<\infty.
    \end{align}
    Moreover, $S_{\mathrm{fb}}\neq\varnothing$ and $\|T(v)\|_W>0$, so
    \begin{align}
        \beta>0.
    \end{align}
    Thus,
    \begin{align}
        0<\alpha\leq\beta<\infty.
    \end{align}

    By the definitions of $\alpha$ and $\beta$, for any $v\in V_{\mathrm{fb}}$, taking
    $v/\|v\|_V\in S_{\mathrm{fb}}$ when $v\neq0$ yields
    \begin{align}
        \alpha\|v\|_V\leq\|T(v)\|_W\leq\beta\|v\|_V.
        \label{eq:app_norm_bounds}
    \end{align}

    We next examine the decomposition of the cost function.
    By the direct-sum decomposition, for every $\eta$,
    \begin{align}
        \delta M(\eta) = \delta M_{\mathcal{K}}+\delta M_{\mathrm{fb}}(\eta)+\delta M_{\mathrm{nc}}
    \end{align}
    uniquely.
    Since $\delta M_{\mathcal{K}}\in\mathcal{K}=\ker T$,
    \begin{align}
        T(\delta M_{\mathcal{K}}) = 0.
    \end{align}
    Moreover, $\delta M_{\mathrm{fb}}(\eta)\in V_{\mathrm{fb}}$ and
    $\delta M_{\mathrm{nc}}\in V_{\mathrm{nc}}$, and therefore, by the definition of $V_{\mathrm{nc}}$,
    \begin{align}
        \langle T(\delta M_{\mathrm{nc}}),T(\delta M_{\mathrm{fb}}(\eta))\rangle_W = 0.
    \end{align}
    Consequently,
    \begin{align}
        C(\eta) &= \|T(\delta M(\eta))\|_W^2 = \|T(\delta M_{\mathrm{nc}})+T(\delta M_{\mathrm{fb}}(\eta))\|_W^2 \notag \\
        &= \|T(\delta M_{\mathrm{nc}})\|_W^2+\|T(\delta M_{\mathrm{fb}}(\eta))\|_W^2.
        \label{eq:app_cost_full_decomposition}
    \end{align}
    Define
    \begin{align}
        C^\ast := \|T(\delta M_{\mathrm{nc}})\|_W^2,\quad C'(\eta) := \|T(\delta M_{\mathrm{fb}}(\eta))\|_W^2.
    \end{align}
    Then
    \begin{align}
        C(\eta) = C^\ast+C'(\eta).
        \label{eq:app_C_Cstar_Cprime}
    \end{align}

    Now define
    \begin{align}
        \gamma := \left(\frac{\alpha}{\beta}\right)^2.
    \end{align}
    Since $0<\alpha\leq\beta<\infty$,
    \begin{align}
        0<\gamma\le1.
    \end{align}
    Suppose that
    \begin{align}
        C(\eta_{\mathrm{final}})<\gamma C(0).
    \end{align}
    Then
    \begin{align}
        & C'(\eta_{\mathrm{final}})<\gamma\left(C'(0)+C^\ast\right)-C^\ast \notag \\
        &= \gamma C'(0)+(\gamma-1)C^\ast\leq \gamma C'(0),
    \end{align}
    and hence
    \begin{align}
        C'(\eta_{\mathrm{final}})<\gamma C'(0).
        \label{eq:app_Cprime_condition}
    \end{align}

    Using Eq.~\eqref{eq:app_norm_bounds},
    \begin{align}
        & \alpha\|\delta M_{\mathrm{fb}}(\eta_{\mathrm{final}})\|_V\leq\|T(\delta M_{\mathrm{fb}}(\eta_{\mathrm{final}}))\|_W \notag \\
        &= \sqrt{C'(\eta_{\mathrm{final}})}<\sqrt{\gamma}\sqrt{C'(0)}.
    \end{align}
    Similarly,
    \begin{align}
        \sqrt{C'(0)} = \|T(\delta M_{\mathrm{fb}}(0))\|_W\leq\beta\|\delta M_{\mathrm{fb}}(0)\|_V.
    \end{align}
    Therefore,
    \begin{align}
        \alpha\|\delta M_{\mathrm{fb}}(\eta_{\mathrm{final}})\|_V<\sqrt{\gamma}\,\beta\|\delta M_{\mathrm{fb}}(0)\|_V = \alpha\|\delta M_{\mathrm{fb}}(0)\|_V.
    \end{align}
    Since $\alpha>0$,
    \begin{align}
        \|\delta M_{\mathrm{fb}}(\eta_{\mathrm{final}})\|_V<\|\delta M_{\mathrm{fb}}(0)\|_V.
    \end{align}
    Since $\mathcal{K}\perp V_{\mathrm{vis}}$, and either $V_{\mathrm{nc}}=0$ or
    $V_{\mathrm{nc}}\perp V_{\mathrm{fb}}$, the error decomposition gives
    \begin{align}
        \|\delta M(\eta)\|_V^2 = \|\delta M_{\mathcal{K}}\|_V^2 + \|\delta M_{\mathrm{fb}}(\eta)\|_V^2 + \|\delta M_{\mathrm{nc}}\|_V^2.
    \end{align}
    Because the $\mathcal{K}$ and $V_{\mathrm{nc}}$ components are independent of $\eta$, it follows that
    \begin{align}
        \|\delta M(\eta_{\mathrm{final}})\|_V < \|\delta M(0)\|_V.
    \end{align}

    Finally, by the first-order expansion of the implemented measurement,
    \begin{align}
        \mathcal{N}\circ\mathcal{F}_{\eta}(\boldsymbol{M}) = \boldsymbol{M}+\delta M(\eta)+O(2),
    \end{align}
    and $\mathcal{F}_0=\mathrm{id}$.
    Therefore,
    \begin{align}
        \left\| \mathcal{N}\circ\mathcal{F}_{\eta_{\mathrm{final}}}(\boldsymbol{M}) - \boldsymbol{M}\right\|_V < \left\| \mathcal{N}(\boldsymbol{M}) - \boldsymbol{M}\right\|_V +O(2).
    \end{align}
\end{proof}

\subsection{Proof of Corollary~\ref{thm:uniform_measurement_improvement}}
\label{app:uniform_measurement_improvement}
\begin{proof}
    For an operator $B\in\mathcal{L}(\mathcal{H})$, we define the Schatten $p$-norm by
    \begin{align}
        \|B\|_p := \left[\operatorname{Tr} \left((B^\dagger B)^{p/2}\right)\right]^{1/p}, \qquad 1\leq p<\infty,
    \end{align}
    and the operator norm by
    \begin{align}
        \|B\|_\infty := \sqrt{\lambda_{\max}(B^\dagger B)}.
    \end{align}
    In particular, $\|\cdot\|_1$ is the trace norm and $\|\cdot\|_2=\|\cdot\|_{\mathrm{HS}}$ is the
    Hilbert--Schmidt norm.

    Consider arbitrary measurement processes $\mathcal{A}=\{A_x\}_{x\in X}$ and
    $\mathcal{M}=\{M_x\}_{x\in X}$, together with an arbitrary $\rho\in S(\mathcal{H})$.
    Since $\mathcal{A}$ and $\mathcal{M}$ are measurement processes, $\|A_x\|_{\infty}\le1$ and
    $\|M_x\|_{\infty}\le1$.

    For each $x\in X$ and every $\rho\in S(\mathcal{H})$,

    \begin{align}
        & \left|p_{\mathcal{A}}(x|\rho)-p_{\mathcal{M}}(x|\rho)\right| \notag \\
        &= \left| \operatorname{Tr} \left[\rho \left(A_x^\dagger A_x-M_x^\dagger M_x\right)\right]\right| \notag \\
        & \leq \left\| \rho \left(A_x^\dagger A_x-M_x^\dagger M_x\right)\right\|_1 \notag \\
        & \leq \|\rho\|_1 \left\| A_x^\dagger A_x-M_x^\dagger M_x\right\|_\infty \notag \\
        &= \left\| A_x^\dagger A_x-M_x^\dagger M_x\right\|_\infty \notag \\
        &= \left\| (A_x-M_x)^\dagger A_x + M_x^\dagger(A_x-M_x)\right\|_{\infty} \notag \\
        & \leq \|(A_x-M_x)^\dagger A_x\|_\infty + \|M_x^\dagger(A_x-M_x)\|_\infty \notag \\
        & \leq \|A_x-M_x\|_\infty\|A_x\|_\infty + \|M_x\|_\infty\|A_x-M_x\|_\infty \notag \\
        &= \left(\|A_x\|_{\infty} + \|M_x\|_{\infty}\right) \|A_x-M_x\|_{\infty} \notag \\
        & \leq 2\|A_x-M_x\|_\infty \notag \\
        & \leq 2\|A_x-M_x\|_2 = 2\|A_x-M_x\|_{\mathrm{HS}}.
        \label{eq:branch_probability_bound}
    \end{align}
    Here we used H\"older's inequality, the triangle inequality, and $\|A\|_\infty\leq\|A\|_2\leq\|A\|_1$.

    Therefore, by the Cauchy--Schwarz inequality,
    \begin{align}
        \frac12 \sum_{x\in X} \left| p_{\mathcal{A}}(x|\rho) - p_{\mathcal{M}}(x|\rho)\right| & \leq \sum_{x\in X} \|A_x-M_x\|_{\mathrm{HS}} \notag \\
        & \leq \sqrt{|X|} \left(\sum_{x\in X} \|A_x-M_x\|_{\mathrm{HS}}^2\right)^{1/2} \notag \\
        &= \sqrt{|X|} \|\mathcal{A}-\mathcal{M}\|_{V}.
        \label{eq:tv_measurement_operator_bound}
    \end{align}

    Since this inequality holds for every $\rho\in S(\mathcal{H})$, it remains valid after taking the supremum
    over $\rho$.
    Substituting $\mathcal{A} = \mathcal{N}\circ \mathcal{F}_{\eta_{\mathrm{final}}}(\mathcal{M})$ and using
    the conclusion of Theorem~\ref{thm:feedback_guarantee},
    \begin{align}
        \left\| \mathcal{N}\circ \mathcal{F}_{\eta_{\mathrm{final}}}(\mathcal{M}) - \mathcal{M}\right\|_{V} < \left\| \mathcal{N}(\mathcal{M}) - \mathcal{M}\right\|_{V} +O(2),
    \end{align}
\end{proof}

\end{document}